\documentclass[11pt]{article}

\usepackage[T1]{fontenc}
\usepackage{lmodern}
\usepackage{hyperref}
\usepackage[letterpaper,margin=1.00in]{geometry}
\usepackage{amsmath, amssymb, amsthm, thmtools, amsfonts}
\usepackage{bbm}
\usepackage{cite}
\usepackage{appendix}
\usepackage{graphicx}
\usepackage{color}
\usepackage{algorithm}
\usepackage[noend]{algpseudocode}
\usepackage{xspace}
\usepackage{mathtools}
\usepackage{todonotes}

\usepackage{booktabs}

\newtheorem{theorem}{Theorem}[section]
\newtheorem{lemma}[theorem]{Lemma}
\newtheorem{claim}[theorem]{Claim}
\newtheorem{remark}[theorem]{Remark}
\newtheorem{corollary}[theorem]{Corollary}

\newtheorem{definition}[theorem]{Definition}

\definecolor{darkgreen}{rgb}{0,0.5,0}
\definecolor{darkblue}{rgb}{0,0,0.5}
\usepackage{hyperref}
\hypersetup{
    unicode=false,
    colorlinks=true,
    linkcolor=darkblue,
    citecolor=darkgreen,
    filecolor=magenta,
    urlcolor=cyan
}
\usepackage[capitalize, nameinlink]{cleveref}

\renewcommand{\paragraph}[1]{\vspace{0.15cm}\noindent {\bf #1}.}

\title{A Loosely Self-stabilizing Protocol for Randomized Congestion Control with Logarithmic Memory\footnote{This is an extended version of a paper which will appear in SSS 2019. This work was partially supported by the German Research Foundation (DFG) within the Collaborative Research Center On-The-Fly Computing (GZ: SFB 901/3) under the project number 160364472.
}
}

\author{
  Michael Feldmann\\
  \small Paderborn University \\
  \small michael.feldmann@upb.de\\
  \and
  Thorsten G\"otte\\
  \small Paderborn University \\
  \small thorsten.goette@upb.de\\
  \and
  Christian Scheideler\\
  \small Paderborn University \\
  \small scheideler@upb.de\\
}

\date{}

\begin{document}

\begin{titlepage}

\maketitle
\thispagestyle{empty}

\begin{abstract}
	We consider congestion control in peer-to-peer distributed systems. 
	The problem can be reduced to the following scenario: Consider a set $V$ of $n$ peers (called \emph{clients} in this paper) that want to send messages to a fixed common peer (called \emph{server} in this paper).
We assume that each client $v \in V$ sends a message with probability $p(v) \in [0,1)$ and the server has a capacity of $\sigma \in \mathbb{N}$, i.e., it can recieve at most $\sigma$ messages per round and excess messages are dropped.
The server can modify these probabilities when clients send messages.
Ideally, we wish to converge to a state with $\sum p(v) = \sigma$ and $p(v) = p(w)$ for all $v,w \in V$.	

We propose a \emph{loosely} self-stabilizing protocol with a slightly relaxed legitimate state.   
Our protocol lets the system converge from \emph{any} initial state to a state where $\sum p(v) \in \left[\sigma \pm \epsilon\right]$ and $|p(v)-p(w)| \in O(\frac{1}{n})$. 
This property is then maintained for $\Omega(n^{\mathfrak{c}})$ rounds in expectation.
In particular, the initial client probabilities and server variables are not necessarily well-defined, i.e., they may have arbitrary values.

Our protocol uses only $O(W + \log n)$ bits of memory where $W$ is length of node identifiers, making it very lightweight.
Finally we state a lower bound on the convergence time an see that our protocol performs asymptotically optimal (up to some polylogarithmic factor).
\end{abstract}

\end{titlepage}

\section{Introduction}
Consider a set of n nodes (called \emph{clients} in this paper) that want to continuously send messages to a fixed node (called \emph{server}) with a certain probability in each round.
The server is not aware of its connections and has limited capabilities with regard to the number of messages it is able to receive in each round and its internal memory.
The task for the server is to use a \emph{congestion control} protocol to modify the client probabilities such that the server receives only a constant amount of messages in each round (on expectation).
As client probabilities may be arbitrary at the beginning, we further require the protocol to be \emph{self-stabilizing}, i.e., it should be able to reach its goal starting from any arbitrary initial state.
Self-stabilization comes with the advantage that the protocol is able to recover from transient faults like message loss or blackout of processes automatically.
As the system grows larger, these kinds of faults occur more often, which makes self-stabilization as a concept very desirable.

At first glance, one may think that this setting only applies to client/server-architectures.
However, we believe that solving this problem is quite important for distributed systems where nodes constantly have to communicate with their neighbors.
Also there are distributed systems where nodes are not aware of their incoming connections, e.g. in rooted trees, random graphs \cite{DBLP:conf/spaa/MahlmannS06} or linearized de Bruijn networks~\cite{DBLP:conf/sss/RichaSS11}.
On these networks one is able to effectively perform many important techniques relevant to distributed computing such as aggregation, sampling, or broadcast which are important for applications like distributed data structures (e.g. hash tables \cite{DBLP:conf/wdag/KniesburgesKS13}, queues \cite{DBLP:conf/ipps/FeldmannSS18} or heaps \cite{DBLP:conf/spaa/0001S19}).
Also nodes with limited capabilities can be found in internet of things applications like wireless networks~\cite{DBLP:journals/tmc/TangGD08}.

In this paper we present a loosely self-stabilizing protocol for congestion control.
In contrast to classical self-stabilization, loose self-stabilization relaxes the closure property.
Our protocol guarantees that the server only receives a constant amount of messages on expectation in each round while only using a logarithmic amount of bits for its internal protocol variables for a period of $O(n^{\mathfrak{c}})$ rounds (and not forever as classical self-stabilization would require).
Furthermore we can guarantee \emph{fairness}, i.e., the probabilities of all clients are the same (up to some small constant deviation).
By slightly weakening the definition for a legitimate state, we are able to analyze the runtime of our protocol and show that it is able to quickly reach a state that is already practical for both, the clients and the server.

\section{Model and Definitions} \label{sec:model}

\subsection{System Model}
\paragraph{Network Model}
Since we only consider communication of nodes with their direct neighborhood in the overlay network, we consider the following directed graph $G=(V \cup \{s\}, E)$.
$V = \{v_1,\ldots,v_n\}$ represents the set of $n$ clients and $s$ represents the server.
We assume $n$ to be fixed.
The set of edges is defined by $E = \{(v,s)\ |\ v \in V\}$, i.e., all clients know the server, but the server does not know which client is connected to it.
More particularly, the server does not know the value $n$.
All clients and the server can be identified via their unique \emph{reference}, represented by values $v_i.id \in \mathbb{N}$ for all $i \in \{1,\ldots,n\}$ and $s.id \in \mathbb{N}$ respectively.
We assume that identifiers can be stored by at most $W$ bits, where $W \geq \log n$ is known to the server.
If a node $v$ knows the reference of another node $w$, then $v$ is allowed to send messages to $w$.

Each client $v \in V$ maintains a probability $p(v) \in (0,\hat{p}]$, where $\hat{p} \leq 1$ is a protocol-specific constant.
Denote by $p_{\mathit{min}} \in (0,\hat{p}]$ the minimum client probability, i.e., $p_{\mathit{min}} = \min_{v \in V}\{p(v)\}$ and denote the sum of all client probabilities by $P$, i.e., $P = \sum_{v \in V}^n p(v)$.
We assume that the probability $p(v)$ for a client $v$ cannot become smaller than $1/2^{b \cdot W}$ for some fixed constant $b > 0$, i.e., it can be encoded by $O(W)$ bits.
This means that all probabilities are multiples of $1/2^{b \cdot W}$.

\paragraph{Computational Model \& Definition of a Round}
We divide time into synchronous \emph{rounds}, where a single round consists of the following steps:

\begin{itemize}
	\item[$(i)$] Each client $v$ tosses a biased coin that shows 'heads' with probability $p(v)$. If $v$'s coin shows 'heads', $v$ sends a message $m = (v.id,p(v))$ to the server $s$. Otherwise $v$ stays idle for the rest of the round.
	We assume that the server is only able to receive up to $\sigma$ messages from clients per round for a fixed constant $\sigma \in \Theta(1)$ that is known to the server.
	If more than $\sigma$ clients decide to send a message to the server in this step, then exactly $\sigma$ of those messages are determined uniformly at random to arrive at the server, while the other ones are dropped.
	\item[$(ii)$] The server makes some internal computation based on the messages it received in the previous step.
	\item[$(iii)$] For each message $m = (v.id,p(v))$ that the server received, it may send a message $m' = (p(v)')$ back to $v$.
	\item[$(iv)$] Each client $v \in V$ that received a message $m = (p(v)')$ in the previous step sets $p(v)$ to $p(v)'$.
\end{itemize}

A message sent by a client to the server in step $(i)$ is denoted as a \emph{ping} or \emph{ping message} and we may also just say that the client \emph{pings} the server in this case.
We say that a client \emph{successfully pings} the server (in round $t$) if it sends a ping message to the server (in round $t$) that is actually being processed by the server, i.e., that is not dropped.
We may use $p_t(v)$ to refer to the probability of client $v$ in round $t$.
Note that the server $s$ is able to answer $v$ in step $(iii)$ because $v$ sent its reference $v.id$ to $s$ in step $(i)$.
Once the round is over, the server forgets about $v.id$.
Also observe that the server is not required to send an answer to each message it received in $(iii)$.

Last, the \emph{state} $S_t$ of the system before round $t$ is defined by the assignment of variables $p(v)$ at each client $v \in V$ and internal variables at the server.
The system transitions from $S_t$ to $S_{t+1}$ by performing the steps $(i)$ to $(iv)$ mentioned above.

\subsection{Problem Statement}
We wish to state a protocol that reaches a state with the following two conditions, namely \emph{Busyness} and \emph{Fairness}. 
They are defined as follows:
\begin{definition}[Busyness] \label{def:business}
	Let $L,R \in O(\sigma)$ be protocol-specific constants.
	We say that the server is \emph{busy} in some state $S$ of the system if $P \in [L,R]$ holds in $S$. We say that a state is $\mathfrak{busy}$ for short.
\end{definition}

\begin{definition}[Fairness] \label{def:fairness}
	The system satisfies \emph{fairness} in some state $S$, if $\sum_{v \in  V}\left(p(v)-\frac{P}{n}\right)^2 \leq \frac{1}{n^c}$ holds in $S$ for some constant $c > 0$. We say that a state is $\mathfrak{fair}$ for short.
\end{definition}
We believe these to be natural and reasonable safety properties given our problem and model setup.  
With the first property we ensure that the server operates close to its limits and is not under- or overutilized.
Note that $L,R \geq 1$ can be chosen freely by the server, so it can adjust these values depending on its computational power in practice.
Note that this is not fully precise in a sense that $P$ does converge to some desired fixed value, but we can guarantee that $P$ will eventually converge to some value within the interval $[L,R]$.
Moreover, the notion of busyness prevents the trivial solution of letting all clients send with probability $1$.
Fairness assures that all clients (roughly) send the same amount of data to the server and every client will eventually send. 
This prevents the trivial solution of letting $\sigma$ clients send with probability $1$ and all others with $0$.

Note that in a distributed setting errors are the norm rather than the exception, which means that the probabilities of the clients and the variables can be corrupted through malicious messages, crashes, and memory faults.
Thus, we are specifically interested in a \emph{self-stabilizing} protocol that reaches a safe state even if all probabilities and server variables are corrupted.

In the classical sense, a protocol is \emph{self-stabilizing} w.r.t. a set of legitimate states if it satisfies \emph{Convergence} and \emph{Closure}: 
Convergence means that the protocol is guaranteed to arrive at a legitimate state in a finite amount of time when starting from an arbitrary initial state.
Closure means that if the protocol is in a legitimate state, it remains in legitimate states thereafter as the set of clients does not change and no faults occur.
However, our protocol will \emph{not} meet these strong requirements of classical self-stabilization due to the clients' probabilistic nature. 
To account for this, we will instead show that our protocol is \emph{loosely self-stabilizing}.

The notion of \emph{probalistic loose self-stabilization} was introduced by Sudo et al. in \cite{DBLP:journals/tcs/SudoNYOKM12} to deal with probabilistic protocols that violate the Closure with \emph{very small} probability.
Instead of a set of legitimate states that are never left, a loosely self-stabilizing protocol maintains a safety condition for a sufficiently long time.
More precisely, a protocol is $(\alpha,\beta)$-loose self-stabilizing, if it fulfills the following two properties: 
First, it reaches a legitimate state after $\alpha$ rounds (in expectation) starting from \emph{any} possible initial state.
Second, given that the execution starts in a legitimate state, the protocol fulfills a safety condition for at least $\beta$ rounds (in expectation).
That means for $\beta$ consecutive rounds, all states fulfill a certain condition if their execution started in a legitimate state.
We call this the holding time.
To put it more formally, let $\mathfrak{S}$ be the set of all possible system states and $\mathfrak{L} \subset \mathfrak{S}$ be the set of all legitimate states.
Then the random variable $C(s,\mathfrak{L})$ denotes the convergence time if the algorithm started in $s \in \mathfrak{S}$.
Likewise, let $\mathfrak{L}^*$ be the set of all states that fulfill the safety condition, then $H(\ell, \mathfrak{L}^*)$ denotes the holding time given that we start in $\ell \in \mathfrak{L}$.
Thus, for a $(\alpha,\beta)$-loose self-stabilizing protocol, it holds
$$
\max_{\mathfrak{s} \in \mathfrak{S}} \mathbb{E}\left[C(\mathfrak{s}, \mathfrak{S})\right] \leq \alpha \,\,\, \textsl{and}
\,\,\, \min_{\ell \in \mathfrak{L}} \mathbb{E}\left[H(\ell, \mathfrak{L}^*)\right] \geq \beta
$$
Note that for an efficient protocol it should hold $\alpha << \beta$, i.e, we quickly reach a legitimate state and then stay safe for a long time.

\subsection{Technical Contributions} \label{sec:intro:problem}
Our goal is to construct a self-stabilizing protocol for the server that converges the system into a state where busyness (\Cref{def:business}) and fairness (\Cref{def:fairness}) hold.
In the following we discuss the most major obstacles that we have to overcome when constructing a solution.

\paragraph{Dealing with Arbitrary Initial States}
In initial states the variables at both the clients and the server may contain arbitrary values.
Particularly, each client probability may initially be an arbitrary value out of $(0,\hat{p}]$.
Due to the restrictions on the message size this may lead to $P$ being as low as $O(1/poly(n))$ initially which means that it may take a long time until the server receives the first ping message.
This means that our protocol needs to be designed in a way such that for initially low values of $P$ we make significant progress in reaching a legitimate state once the probability of a client is modified.

\paragraph{Knowledge of $\Theta(\log n)$}
Our algorithm requires the server to estimate $\Theta(\log n)$.
The problem of approximating $\Theta(\log n)$ can be non-trivial when additionally requiring a self-stabilizing solution for this, i.e., the server may think of any value to be $\log n$ initially.
Our loosely self-stabilizing solution for approximating $\Theta(\log n)$ at the server may be of independent interest.

\subsection{Our Contribution}
We propose a congestion control protocol that is loosely self-stabilizing. 
It converges to a legitimate state that is $\mathfrak{busy}$ and $\mathfrak{fair}$ within $\Tilde{O}(\mathfrak{c}\left(p_{min}^{-1}+ n^3)\right)$ \footnote{We use $\tilde{O}$ to hide polylogarithmic factors.} rounds starting from any initial state where clients may have arbitrary probabilities.
Then all following states are also $\mathfrak{busy}$ and $\mathfrak{fair}$ for at least another $O(n^{\mathfrak{c}})$ rounds in expectation. 
Here, $\mathfrak{c}$ is a parameter and can be chosen depending on the context.
Note that even for small $\mathfrak{c}$ the system stays stable long enough for practical purposes. 
Furthermore, the server uses only $O(W + \log n)$ bits in legitimate states.
This makes the protocol very lightweight and ideal for servers with strong memory constraints, e.g., in sensor networks.

The rest of the paper is structured as follows:
First, we review some related work in \Cref{sec:related_work}.
Then, we present our protocol in \Cref{sec:desc}. 
Last, in \Cref{sec:correctness_analysis} we rigorously analyze our protocol and show that it is loosely self-stabilizing.

\section{Related Work} \label{sec:related_work}
\paragraph{Congestion Control}
There exists a wealth of literature on \emph{congestion control} in the internet.
Classical approaches that have been considered are MIMD~(Multiplicative Increase, Multiplicative Decrease~\cite{DBLP:journals/ccr/Kelly03}) and AIMD (Additive Increase, Multiplicative Decrease~\cite{DBLP:journals/cn/ChiuJ89}).
Many other researchers studied congestion control for the AIMD model, which resulted in various extensions of the original work, see for example~\cite{DBLP:journals/automatica/CorlessS12}, \cite{DBLP:journals/algorithmica/KesselmanM05}, \cite{DBLP:conf/wwic/LahanasT05}.
Although these protocols work for arbitrary initial probabilities, their auxiliary variables are always assumed to be well-initialized.
In contrast, our protocol also tolerates completely arbitrary initial states including auxiliary variables, making it truly self-stabilizing.
Also, to the best of our knowledge, prior congestion control protocols do not provide a rigorous theoretical analysis on their convergence time.

\paragraph{Flow Control}
Close to congestion control problems are flow control problems (see \cite{1094691} for a survey).
These protocols differ from our setting in the sense that they operate on a continuous data stream, whereas we consider discrete rounds where only small self-contained control-messages are exchanged between the server and multiple clients, so flow control strategies are not applicable here.

\paragraph{Contention Resolution}
Close but different to congestion control protocols is the area of contention resolution in multiple access channels (see for example~\cite{DBLP:conf/spaa/BenderFHKL05}, \cite{DBLP:journals/jacm/BenderFGY19}, \cite{DBLP:conf/soda/ChangJP19} or~\cite{LAG02} for a survey).
A multiple access channel~(MAC) is a medium shared among all nodes through which they can send messages.
In each round a node may either send a message or sense the channel.
Messages that have been sent in the same round by two or more nodes \emph{collide} and are not transmitted.
By sensing the channel a node gets informed whether the channel is \emph{idle} (no message has been sent), \emph{busy} (a collision occurred) or it receives a message (in case there has been exactly one message sent).
Contention resolution differs from congestion control in a sense that once two or more messages are sent in the same round there already is a collision, whereas in congestion control multiple messages are allowed to be processed by the receiver.
Also the MAC allows clients to only receive binary feedback, making it less powerful compared to our server.

\paragraph{Distributed Consensus and Load Balancing}
Further related areas on a technical level are \emph{distributed average consensus} (see~\cite{DBLP:conf/ac/GuerraouiHMORS99} for a survey) and (discrete) load balancing (see \cite{berenbrink},\cite{DBLP:conf/icalp/TalwarW14} and the references therein). In both problems, multiple agents try to find the arithmetic mean of a given set of initial values.
Our protocol tries the same in order to achieve fairness.
However, we need to deal with dynamically changing probabilities as 
the adaption of the nodes' values directly influences their sampling  probabilities. 
In other settings the probabilities may be arbitrary but are fixed in advance. 

\paragraph{Self-stabilization}
Self-stabilization was first proposed in \cite{DBLP:journals/cacm/Dijkstra74}.
Since inventing self-stabilizing protocols can be quite difficult, people came up with relaxed versions for the convergence property like probabilistic self-stabilization or weak-stabilization~\cite{DBLP:journals/ijfcs/DevismesTY15}.
The notion of loose-stabilization \cite{DBLP:journals/tcs/SudoNYOKM12} that is used in this paper relaxes the closure property instead of the convergence property.

\section{Protocol Description} \label{sec:desc}
Intuitively our protocol works as follows: We constantly let the server count the number of pings it received in each round for an interval of $\Delta$ rounds.
Probabilities of clients that ping are averaged in these rounds.
Once an interval of $\Delta$ rounds ends, the server is able to precisely approximate $P$ in case $\Delta \in \Theta(\log n)$ and decide whether to either raise the probability of a client that has pinged in that round (if $P$ is too small), decrease the probability of a client (if $P$ is too large) or adjust the probabilities of clients by computing the average (if $P$ lies within a desired interval).

We describe the protocol in greater detail now starting with the introduction of variables and constants.
Afterwards we describe how the approximation for $P$ at the server works, followed by the description of the core protocol.
We refer the reader to \Cref{app:approx_log_n} where we describe how to obtain an estimation of $\Theta(\log n)$ in a self-stabilizing manner (which may be of independent interest) that is then stored in $\Delta$.

\subsection{Variables and Constants}
\Cref{table:variables} shows the variables and constants that are maintained by the server.

\begin{table*}[ht]
\centering
\begin{tabular}{@{}lp{13.5cm}@{}}
\toprule 
 $\varepsilon > 0$ & A constant used for the approximation of $P$.\\
 $L,R \in \Theta(1)$& Constants for the left and right border of the desired interval $[L,R]$ to which $P$ should converge. 	In order to guarantee that eventually $P \in [L,R]$, we require that $|R-L| > \hat{p} + 2\varepsilon$.
Note that $L,R$ are chosen such that $1 \leq L < R \leq \sigma$, i.e., on expectation, the server receives at least $L$, but no more than $R$ messages in legitimate states.\\
$\Delta \in \Theta(\log n)$& A variable indicating the interval of rounds in which the server counts the number of incoming pings.\\
$\delta \in [0,\Delta]$& A counter that is incremented each round and reset to $0$ once it is equal to $\Delta$.\\
$X \in \mathbb{N}_0$& A counter that sums up the number of incoming pings within a period of $\Delta$ rounds.\\
\bottomrule
\end{tabular}
\caption{Variables and constants used by our algorithm}
\label{table:variables}
\end{table*}

Note that the constants $L, R$ and $\varepsilon$ are protocol-based constants, which means they are chosen preemptively by the server and thus are fixed while the stabilization process of the system is going on.
On the other side the variables $\delta, \Delta$ and $X$ may contain arbitrary values out of their domains in initial states.

\subsection{Approximating \texorpdfstring{$\Theta(\log n)$}{Theta(log n)} at the Server}
In order to work properly, our protocol needs an approximation of $\Theta(\log n)$. 
In the following we sketch a protocol to obtain such an approximation given that we have one server and $n$ clients.

We let the server maintain a table of $\log \log N$ columns where each column $i$ represents a value $c_i = \sqrt[2^i]{N}$ and a timestamp $t_i \geq 0$ (see \Cref{table:log_n}).
	The first column $c_0$ represents the value $N$, which may be arbitrary large in initial states.
	Therefore the table along with its timestamps may initially be completely arbitrary.
	\begin{table}[ht]
		\centering
		\begin{tabular}{|c|c|c|c|c|c|}
			\hline
				$c_0 = N$ & $c_1 = \sqrt{N}$ & $c_2 = \sqrt[4]{N}$ & ... & $c_{\log \log N - 1} = 2$\\ \hline
				$t_0$ &  $t_1$ & $t_2$ & ... & $t_{\log \log N - 1}$ \\ \hline
		\end{tabular}
		\caption{Table maintained at the server.}
		\label{table:log_n}
	\end{table}
	The table is maintained as follows by the server: We map the identifiers of the server and the clients to the interval $[0,1)$ via a uniform hash function $h: \mathbb{N} \rightarrow [0,1)$.
	Whenever a client $v$ with $|h(s.id) - h(v.id)| \leq \frac{1}{c_i}$ successfully pings the server, the server resets all timestamps $t_i,\ldots,t_{\log \log N -1}$ to $0$.
	Aside from this, each timestamp $t_i$ gets incremented by one in each round.
Once the entry $t_i$ for column $c_i$ gets larger than $O(c_i \cdot polylog(c_i))$, all columns $c_0,\ldots,c_i$ are deleted from the table and the value $N$ is set to the column $c_{i+1}$.
	On the other side, once a client $v$ pings for which $|h(s.id) - h(v.id)| \leq \frac{1}{c_0^2}$ holds we update the table by adding that many columns to the left until $\frac{1}{c_0^2} < |h(s.id) - h(v.id)| \leq \frac{1}{c_0}$ holds.
	The server always sets $\Delta = \Theta(\log c_0)$ to approximate $\Theta(\log n)$.

This protocol will run in parallel to anything described in the remainder of this section.

\subsection{Approximating \textit{P} at the Server}
At the end of an interval of rounds of size $\Delta$, the server checks whether $P$ is (approximately) less than $L$, larger than $R$ or within $[L,R]$.
We use the operator $\prec$ to indicate the result of the approximation, for example if $P$ is approximately less than $L$ we say $P \prec L$ and otherwise $P \succ L$.
In order to check whether $P \prec L$ or $P \succ L$, the server checks whether $X/\Delta < L$ holds.
If that is the case then the server decides on $P \prec L$, otherwise it decides $P \succ L$.
By comparing $X/\Delta$ to $R$ the server can do the same to decide whether $P \prec R$ or $P \succ R$ holds. 

\subsection{Core Protocol}
The server executes \Cref{algo:protocol} in each round after each client has decided whether to ping the server or not (\Cref{algo:client_protocol}, \Cref{algo:client_protocol:1}).
Here $v_1,\ldots,v_k$ are the clients that successfully pinged the server in round $t$.

\begin{algorithm}
\caption{Pseudocode executed at each client $v$ in each round}
\label{algo:client_protocol}
\begin{algorithmic}[1]
	\State Toss a coin that shows 'heads' with probability $p(v)$
	\If{Coin shows 'heads'}
		\State Send $m = (v.id, p(v))$ to $s$ \label{algo:client_protocol:1}
	\EndIf
	\If{$v$ received $p'(v)$ from $s$}
		\State $p(v) \gets p'(v)$
	\EndIf
\end{algorithmic}
\end{algorithm}

\begin{algorithm}[ht]
\caption{Pseudocode executed at the server in each round}
\label{algo:protocol}
\begin{algorithmic}[1]
	\State Let $v_1,\ldots,v_k$ be the clients that successfully pinged the server  in ascending order of their probabilities, i.e., $p(v_1) \leq \ldots \leq p(v_k)$\label{algo:line1}
	\State $X \gets X + k$ \label{algo:update_X}
	\State $\delta \gets (\delta + 1) \mod \Delta$ \label{algo:update_delta}
	\If{$\delta = 0$}
		\If{$P \prec L$}
			\State Send $\hat{p}$ to $v_1$ \Comment{Increase minimum probability} \label{algo:p_less_l} 
		\ElsIf{$P \succ R$ and $k \geq 2$}
			\State Send $p(v_k)/(1+1/\sigma)$ to $v_k$ \Comment{Decrease maximum probability}\label{algo:average_rule:1} 
		\EndIf
	\State $X \gets 0$ \label{algo:reset_X}
	\Else
		\ForAll{$i \in \{1,\ldots,k\}$}
		\State Send $\lfloor \sum_{i=1}^k p(v_i)/k) \rfloor + r_i$ to $v_i$ \Comment{Average probailities} \label{algo:average_rule:2}
		\EndFor
	\EndIf
\end{algorithmic}
\end{algorithm}

The protocol given by \Cref{algo:protocol} works as follows: At the beginning of each round we let clients ping the server with their corresponding probabilities.
Assume that $k$ clients $v_1,\ldots,v_k$ pinged the server ordered by their probabilities, i.e., $p(v_1) \leq \ldots \leq p(v_k)$.
The server first increments $X$ by $k$ (\Cref{algo:update_X}) and then sets $\delta$ to $(\delta + 1) \mod \Delta$ (\Cref{algo:update_delta}).
In case $\delta \neq 0$, the server sets each probability $p \in \{p(v_1),\ldots,p(v_k)\}$ to the average of these probabilities (\Cref{algo:average_rule:2}).
In a round where $\delta = 0$ holds the server instead approximates $P$ based on $X$ and $\Delta$.
Using the approximation for $P$, the server checks whether $P \prec L$, i.e., whether $P$ is currently too low.
If that is the case, then the server raises the minimum probability $p(v_1)$ to $\hat{p}$ (\Cref{algo:p_less_l}).
On the other hand, if $P$ is too large ($P \succ R$) and at least $k \geq 2$ clients pinged, the server sets the maximum probability $p(v_k)$ to $p(v_k)/(1+1/\sigma)$ (\Cref{algo:average_rule:1}).
Once this has been done, the server resets $X$ to $0$ (\Cref{algo:reset_X}).

Notice that parts of our algorithm (specifically the way we choose client probabilities to be decreased) are related to the well-known \emph{two-choice process} where we (greedily) choose the process with minimum probability to have its probability reduced (\Cref{algo:average_rule:1}).
As it turns out in the analysis, we can make use of this by modelling our setting as a balls-and-bins process for which we can apply a result from \cite{DBLP:conf/icalp/TalwarW14}.

Due to messages being restricted to only $O(W)$ bits it may happen that we lose accuracy on the overall sum of probabilities $P$ if we were to simply compute the averages of client probabilities and round it up or down.
To overcome this problem, we use the following rounding approach when computing average client probabilities (\Cref{algo:average_rule:2}): In a round where $k$ clients ping the server and the average of these clients has to be computed, we initially set the probabilities to the average rounded down on $W$ bits, i.e., the least significant bit is set to $0$.
As the real average value leaves some residue value of the form $r \cdot \frac{1}{2^{b \cdot W}}$ for an integer $r < k$, we set the least significant bit of $r$ clients (chosen randomly among the $v_i$'s) to $1$.
This is indicated by the values $r_i \in \{0,\frac{1}{2^{b \cdot W}}\}$.
By doing so we ensure that $P$ does not get modified when only computing averages and all the client probabilities remain multiples of $\frac{1}{2^{b \cdot W}}$.
For the analysis we assume for simplicity that we compute the average value without rounding and only consider the rounding approach when it actually influences a proof.

\section{Analysis} \label{sec:correctness_analysis}
We analyze our algorithm in this section and show that it is loosely self-stabilizing.
Therefore, we need to give a formal definition for a legitimate state and a safety condition.
Obviously, we want our system to be in a \emph{busy} and \emph{fair} state, but moreover, 
in order to guarantee a long holding time, we need a correct estimate of $\Theta(\log n)$. 
Therefore, we introduce the notion of stability.
\begin{definition}[Stability]
A state $s \in S$ fulfills the stability property, if $c_0$, the biggest entry in the table, is in $\Omega(n^{\frac{1}{2}})$ and all $t_i$ are $0$. 
We call such a state $s$ $\mathfrak{stable}$ for short. 
\end{definition}
As we will see, this ensures that the protocol correctly estimates $\Theta(\log n)$ for at least $\Omega(n^{\mathfrak{c}})$ rounds in expectation.

Furthermore, we need to weaken the fairness property a bit to get more practical results.
This comes from the fact that the algorithm may erroneously increase or decrease the probabilities, even if $\Delta \in \Theta(\log n)$.
We wish to acknowledge that our protocol does reach an arbitrarily fair state after $O(poly(n))$ rounds and then stays that way for another $O(poly(n))$ rounds (both in expectation), i.e., it would hold $\alpha \approx \beta$.
We sketch this in \Cref{app:correctness_analysis}.
We therefore focus on the so-called weakly fair state as we deem it more practical.
It is defined as follows:
\begin{definition}[Weakly Fairness] \label{def:weak_legal_state}
	A state $S$ of the system is a $\mathfrak{weakly fair}$ state if  $\forall v \in V: p(v) \in \Omega \left(\frac{P}{n}\right)$. 
\end{definition}
Given this definition, we can now simply define the legitimate state.
Over the course of this chapter, we will show that the following holds:
\begin{theorem} \label{theorem:overall_convergence_time}
Let $\mathfrak{c}$ be a big enough constant.
Further, let $L,R \in O(\sigma)$ and $\varepsilon > 0$ be protocol-specific constants. Then it holds:
\begin{itemize}
    \item A state $\ell \in \mathfrak{L}(L,R,\varepsilon)$ of the system is a \emph{legitimate state} if it is $\mathfrak{busy}$, $\mathfrak{weakly fair}$, and $\mathfrak{stable}$.
    \item A state $\ell \in \mathfrak{L}^*(L,R,\varepsilon)$ of the system furfills the safety condition if it is $\mathfrak{busy}$ and $\mathfrak{weakly fair}$.
\end{itemize}
Then, our protocol is $\left(\Tilde{O}(p_{min}^{1}+n^3),\Omega(n^{\mathfrak{c}})\right)$-loosely self-stabilizing with regard to the legal states $\mathfrak{L}(L,R,\varepsilon)$ and safe states $\mathfrak{L}^*(L,R,\varepsilon)$.
\end{theorem}

\subsection{Convergence Time} \label{sec:analysis:convergence_time}
Now we show that the system converges to a legitimate state after $\tilde{O}(p_{\mathit{min}}^{-1} + n^3)$ rounds w.h.p.
We split the analysis into three phases: First we analyze the time it takes until $\Delta \in \Theta(\log n)$ is fixed.
In the second phase we analyze the time it takes for $P$ to reach a value within $[L,R]$.
Finally we show a bound on the time it takes until weak fairness is reached, i.e., until all probabilities are in $\Omega(P/n)$.
The full proofs are deferred to \Cref{app:approx_log_n}, \Cref{app:P_convergence_time} and \Cref{app:fairness_complete}.
Note that these phases exist purely for analytical purposes and the algorithm itself is oblivious of them.

\subsubsection*{Phase I: Approximating $\Theta(\log n)$}
We start by showing that there exists a appropriate self-stabilizing approximation algorithm for $\Theta(\log n)$ given that the communication graph is a star graph of $\Theta(n)$ nodes. 
In particular, the following holds:

\begin{theorem} \label{theorem:log_n:approx}
Our protocol provides a fixed estimation of $\Theta(\log n)$ for the server within $O(p_{\mathit{min}}^{-1} + n^2 \cdot polylog(n))$ rounds w.h.p. starting from any configuration, and reaches a $\mathfrak{stable}$ state every $O(n^2)$ rounds with probability $1-o(n^{-\mathfrak{c}})$.
\end{theorem}

\begin{proof}[Proof Sketch] For the analysis of this approach we first show that after $O(n^2 \cdot polylog(n))$ rounds all superfluous columns that may exist in initial states have been deleted and thus $\Delta \leq \Theta(\log n)$ holds.
	Afterwards we show that after $\Theta(n/\log n)$ clients have successfully pinged the server at least once (which needs $O(p_{\mathit{min}}^{-1} +  n \cdot \log^2 n)$ rounds, see the analysis in \Cref{app:sec:n_pings}), at least $\Theta(n/\log n)$ clients are \emph{visible}, i.e., they have a probability of at least $\Omega\left(\frac{P}{n \cdot polylog(n)}\right)$.
	This suffices to show convergence for our strategy.
	
	For the second property we show that no columns gets added or deleted w.h.p. and that a visible client remains visible throughout the algorithm via a slight adaptation of the server's behavior.
	This leads to the timestamp $t_0$ of the first column $c_0$ being reset to $0$ after at most $O(n^2)$ rounds w.h.p.
\end{proof}

\subsubsection*{Phase II: Convergence for $P$}

In the following we bound the time until we arrive at a configuration with $P \in [L,R]$ once $\Delta \in \Theta(\log n)$ has stabilized.
Here, we need to take into account that in the first phase all probabilities could be arbitrarily adapted by the algorithm.
In particular, through negative feedback the smallest probability $p_{\mathit{min}}$ could be further reduced.
This could potentially delay the stabilization of our algorithm ad infinitum.
However, recall that the minimal probability is only decreased when two nodes of (almost) minimal probability successfully ping the server.
Thus, the smaller $p_{\mathit{min}}$ gets, the more unlikely it is for it to be reduced further.

Formally, we can show the following:

\begin{lemma}\label{lemma:negative_feedback}
During the execution of the first phase, no node will be assigned a probability smaller than $O\left(\frac{\min\{p_{\mathit{min}},n^{-2}\}}{\log{n}} \right)$ w.h.p.
\end{lemma}

\begin{proof}[Proof Sketch] The proof works similar to the analysis of a ball-into-bins process with $d$ choices. 
Whenever the probabilities are reduced through the algorithm, this can be seen as throwing a ball to the biggest of the $d$ randomly chosen nodes that pinged in that round.
Through a careful adaption of the corresponding proof, we see that the minimal node's probability is reduced at most $\log{\log{n}}$ times if the protocol runs for $O(p_{min}^{-1})$ rounds.
This corresponds to reducing the probability by a factor $\left({1+\frac{1}{\sigma}}\right)^{-\log{\log{n}}}$. Since $\sigma$ is constant, this is within  $O\left(\log{n}^{-1}\right)$. 
We defer the full proof to \Cref{app:negative_feedback}.
\end{proof}
Given this insight, we can now show the following:
\begin{theorem}[Convergence Time for $P$] \label{theorem:P_convergence_time}
	Let the system be in any state where $\Delta \in \Theta(\log n)$ is already fixed.
	After $O((p_{\mathit{min}}^{-1} + n) \log^2 n)$ rounds, the system reaches a state where $P \in [L,R]$ w.h.p.
\end{theorem}

\begin{proof}[Proof Sketch]
	We need to consider the cases $P < L$ and $P > R$.
	In case $P < L$ we can show that it takes $O(p_{\mathit{min}}^{-1} \log n)$ rounds until $P \in [L,R]$ w.h.p.
	This follows from the time needed to set the probabilities of at least $\alpha$ different clients to $\hat{p}$ for a constant $\alpha \in \mathbb{N}$ with $\alpha \cdot \hat{p} > L$.
	For $P > R$ we can conclude that, with at least constant probability, at least one client out of the set $V' = \{v \in V\ |\ p(v) \geq \frac{P}{2n}\}$ pings the server successfully in a round where $\delta = 0$ and thus gets its probability reduced.
	It follows by calculation that after $O(n \log n)$ of these reductions $P < R$ holds.
	These reductions can be achieved within $O(n \log^2 n)$ rounds w.h.p.
\end{proof}

\subsubsection*{Phase III: Convergence to Weak Fairness}

Finally, we show that we reach (weakly) fair state after at most $\tilde{O}(p^{-1}_{\mathit{min}} + n^3)$ rounds w.h.p. given that the initial state is already $\mathfrak{busy}$ and $\mathfrak{stable}$.
Our definition of a weakly fair state requires that all probabilities are close to $\frac{P}{n}$ (and moreover will \emph{stay} close for $O(n^{\mathfrak{c}})$ rounds).
To achieve this, the protocol must not in- or decrease the client probabilities too often.
On the first glance, one might think that $\Delta \in \Theta(\log n)$ and $P \in [L,R]$ are sufficient to ensure that.
However, a closer look reveals that in cases where $P$ is close to the borders of the interval $[L,R]$ this might not be the case.

However, if we assume that $P$ only changes very infrequent, then we can adapt the results of Berenbrink et. al. \cite{berenbrink} and obtain the following result (the adapted proof can be found in \Cref{app:fairness_complete}).

\begin{theorem} \label{theorem:deckel}
	Let the system be in a legitimate state where $P \in [L,R]$.
	Then it holds:
	\begin{enumerate}
		\item After at most $O(n^3\log n)$ rounds, $P$ changes only with prob $o(\frac{1}{n^2})$. 
			\item After $O(p_{min}^{-1} \cdot \log n)$ rounds the system reaches a $\mathfrak{weakly fair}$ state w.h.p.
	\end{enumerate}		
\end{theorem}

\begin{proof}[Proof Sketch]
	For the first claim, we present a simple technical argument in the appendix.
	Note that the probability to decrease $P$ depends on $P$ itself.
	The main idea is that after $n$ reductions (which take $n^3$ rounds in expectation), $P$ is so small that further reductions are very unlikely. 
For the second part, we model our system as a balls-and-bins process.
The clients represent the bins and the client probabilities represent balls, where the number of balls depends on $P$, i.e., if the probability of a client $v$ is $p(v) = \frac{c}{2^{b \cdot W}}$ (recall that client probabilities are multiples of $1/2^{b \cdot W}$), then we say that $v$ has $c$ balls.
	At the beginning the $P \cdot 2^{b \cdot W}$ balls are arbitrarily distributed among all clients.
	Then, we use an adaption of the potential function analysis from \cite{berenbrink}. 
	As potential, we/they use the sum of squared differences, i.e.,
	$
		\Phi_t := \sum_{v \in V} \left(p_t(v) - Pn^{-1}\right)^2
	$.
	In particular, we need to adapt the following:
	\begin{enumerate}
	    \item The probabilities are not uniform \emph{and} change dynamically during the process. 
	    We solve this by observing that with constant probability, the sampled values are close to the arithmetic mean.
	    Therefore, clients with small probability are increased quickly once they send.
	    \item The probabilities can be reduced. However, since the probability for a change is small, i.e., $o(n^{-2})$, we can amortize it through the balancing.
	\end{enumerate}
	Given these adaptations, we can show that after $O(p^{-1}_{min} \cdot \log n)$ rounds the sum of the squared differences between all clients and the average is at most $n$. This corresponds to the sum of the squared differences between all client probabilities and the average probability $P/n$ being at most $1/n$, which suffices to show fairness. 
Together with the time it takes for probabilities to be small enough, the theorem follows.
\end{proof}

\subsection{Holding Time}
It remains to bound the holding time.
However, this simply follows from the observations we made so far.
\begin{theorem}
Let $\mathfrak{c}$ be an arbitrary constant.
Suppose the system is in a legitimate state $\ell \in \mathfrak{L}(L,R,\epsilon)$, then it will remain in a safe state for $\Omega(n^{\mathfrak{c}})$ rounds in expectation.
\end{theorem}
\begin{proof}[Proof Sketch]
We show that both busyness and fairness are maintained with probability $1-o(n^{\mathfrak{c}})$ if we start in a $\mathfrak{stable}$ state. 
First, note that starting in a $\mathfrak{stable}$ state, the system maintains $\Delta \in \Theta(\log{n})$ until the first entry in the table is deleted.
For a deletion, a node (which pings with probability $\frac{P}{n}$) must not ping for consecutive $O(\mathfrak{c} n \log n)$ rounds.
The probability for this is $O(n^{-\mathfrak{c}})$ and hence this holds for $O(n^{\mathfrak{c}})$ round in expectation.
Given that $\Delta \in \Theta(\log{n})$ remains fixed, we can show the following.
\begin{enumerate}
\item The system remains in a busy state. 
We violate busyness if and only if $P$ leaves the interval $[L,R]$.
Therefore, the probabilities need to be de- or increased at least $\omega(n)$ times. 
This only happens if the server (wrongly) predicts $P\succ L$ or $P \prec R$, which happens with prob. $1-o(n^{\mathfrak{c}})$ given that $\Delta \in \Theta(\log{n})$.
\item The system remains in a weakly fair state. 
By a similar argument, we see that fairness is violated if few nodes are decreased too often. 
This also only happens if the server (wrongly) predicts $P\succ L$ or $P \prec R$ and is therefore evenly unlikely.
In particular, the times between two decreases are so long that the nodes can re-balance themselves and thus stay $\mathfrak{weakly fair}$.
\end{enumerate}
\end{proof}
\subsection{Tightness}
Last, we observe the tightness of our convergence time.
One can easily see that \emph{any} self-stabilizing protocol needs $\Omega(p_{\mathit{min}}^{-1} \log n + n)$ rounds to reach a legitimate state (the full proof is deferred to \Cref{app:lower_bound}).

This follows from the fact that each client need to ping the server at least once to get a probability in $O(\frac{P}{n})$.
As we see, our protocol is indeed optimal if $p_{min} \in O(\frac{1}{n^3})$, but is slower otherwise. 
However, note that the slowdown only happens because of two important properties that our protocol fulfills.
First, it takes an additional $O(n^3)$ rounds in phase $I$, i.e., during approximation of $\Theta(\log n)$.
Given that the protocol has a stable estimation of $\Theta(\log{n})$ (which is reasonable in many contexts) the convergence time is asymptotically optimal in this phase.
Second, it takes $O(n^3)$ rounds until the probability for a decrease is so low that the protocol converges to a weakly fair state.
For an even notion of fairness (e.g. at most $o(n)$ nodes may have very low probability) this could be improved.

\section{Conclusion} \label{sec:conclusion}
We proposed a self-stabilizing protocol for congestion control in overlay networks that performs reasonably well in our model.
Finally we want to make a remark on the system's performance in arbitrary topologies.

\begin{remark}
	Consider an overlay network $G = (V,E)$ with indegree at most $\zeta$. Further, let each node know a (probably rough) estimation $N$ of $n$.
	Assume we apply our protocol for loose self-stabilizing congestion control such that each node acts as a server for its incoming connections and as a (separate) client for each of its outgoing connections.
	This way we obtain ($O((p_{\mathit{min}}^{-1} + \zeta^3) \cdot polylog(N))$,$N^{\mathfrak{c}}$) loosely self-stabilizing protocols for all servers.
\end{remark}

This follows from \Cref{theorem:overall_convergence_time},
if we assume that all nodes $v \in V$ run our algorithm with neighbors as clients.
However note that in cases where $\zeta$ is constant our results would hold only with probability in $\Theta(e^{-\zeta})$ and not w.h.p.
To circumvent this we just use the estimation $N$ of $n$ for nodes and let the value for $\Delta$ at each node be in $\Theta(\log N)$ instead of $\Theta(\log \zeta)$.
Given that all nodes know $\Theta(\log N)$, the approximation algorithm is obsolete and all states are $\mathfrak{stable}$.
Note that all other bounds only depend on the number of client.
Thus, we plug in the maximum degree $\zeta$ of a node instead of $n$.
This gives us the $polylog(N)$-factor in the runtime above.

\bibliographystyle{halpha-abbrv}
\bibliography{literature}

\newpage

\appendix

\section{General Notions from Probability Theory} \label{sec:preliminaries}
In this section we introduce some well-known facts in probability theory.
We may use $\exp(x)$ to denote $e^x$.

We make extensive use of the following Chernoff bounds:

\begin{lemma}[Chernoff Bounds] \label{lemma:chernoff}
	Let $X_1,\ldots,X_n$ be a set of independent binary random variables.
	Let $X = \sum_{i=1}^n X_i$ and $\mu = \mathbb{E}[X]$.
	Then it holds:
	\begin{itemize}
		\item[(i)] For any $\delta \geq 1$ it holds \[\Pr[X > (1+\delta)\mu] \leq \exp\left(\frac{-\delta \mu}{3}\right).\] 
		\item[(ii)] For any $0 \leq \delta \leq 1$ it holds \[\Pr[X > (1+\delta)\mu] \leq \exp\left(\frac{-\delta^2 \mu}{3}\right).\] 
		\item[(iii)] For any $0 \leq \delta \leq 1$ it holds \[\Pr[X \leq (1-\delta)\mu] \leq \exp\left(\frac{-\delta^2 \mu}{2}\right).\] 
	\end{itemize}		
\end{lemma}

We also use the following generalization of \Cref{lemma:chernoff}$(iii)$:

\begin{lemma}[Hoeffding Bound] \label{lemma:hoeffding}
	Let $X_1,\ldots,X_n$ be a set of independent random variables with $X_i \in [0,b]$.
	Let $X = \sum_{i=1}^n X_i$ and $\mu = \mathbb{E}[X]$.
	Then for any $0 \leq \delta \leq 1$ it holds \[\Pr[X \leq (1-\delta)\mu] \leq \exp\left(\frac{-\delta^2 \mu}{2b}\right).\] 
\end{lemma}

Besides the Chernoff Bound, we will also use another standard bound, namely Markov's inequality.
It is defined as follows:
\begin{lemma}[Markov's Inequality] \label{lemma:markov}
Let $X \geq 0$ be non-negative random variable and $a > 0$ be a constant.
Then it holds
\[
	\Pr\left[X > a\mathbb{E}[X]\right] \leq \frac{1}{a}.
\]
\end{lemma}

\section{Strategies for Estimating \texorpdfstring{$\log n$}{log n}} \label{app:approx_log_n}
We provide the details on how the server is able to get an approximation for $\Theta(\log n)$ that is stored in the variable $\Delta$ at the server.
As pointed out in \Cref{sec:intro:problem} this cannot be done trivially in a self-stabilizing fashion.
Still, there are multiple ways of doing this that come with both advantages and disadvantages.
We want to present and discuss two strategies in this section.

\subsection{Strategy 1: Storing a Single Client Identifier} \label{app:approx_log_n:1}
We first introduce the following simple strategy.

\paragraph{Description}
We map the identifiers of the server and the clients to the interval $[0,1)$ via a uniform hash function $h: \mathbb{N} \rightarrow [0,1)$.
The server $s$ maintains an additional variable $\hat{v} \in \mathbb{N}$ that stores the identifier of the client $v$ that minimizes $|h(s.id) - h(v.id)|$.
$\hat{v}$ is constantly updated whenever a client pings $s$, i.e., if a client $v$ pings $s$, then $s$ checks whether $|h(s.id) - h(\hat{v})| < |h(s.id)-h(v.id)|$ and updates $\hat{v}$ if necessary.
Whenever $s$ updates $\hat{v}$, it also updates $\Delta$ by setting $\Delta = -\log |h(s.id) - h(\hat{v})| \cdot c$, where $c$ is a protocol-specific constant (see \Cref{lemma:approx_correctness}).
We show that $-\log |h(s.id) - h(\hat{v})| \in \Theta(\log n)$ w.h.p. once at least $n/\log n$ clients have successfully pinged the server.

\paragraph{Correctness}
It is easy to see that if all clients have pinged the server at least once, then $\hat{v}$ stores the identifier of the client $v$ that minimizes $|h(s.id) - h(v.id)|$.
Also $\hat{v}$ does not change from this point on, since the set of clients does not change.
Now we show that the value $\Delta$ is a correct estimate of $\log n$ once $n/\log n$ clients have successfully pinged the server.
For this we need the following lemma from~\cite{DBLP:conf/podc/MalkhiNR02}:

\begin{lemma}[\hspace{1sp}\cite{DBLP:conf/podc/MalkhiNR02}] \label{lemma:distance_approximation}
	Given $n$ clients $v_1,\ldots,v_n$ and a server $s$.
	Let $\hat{v}$ store the identifier of the client $v \in V$ such that $|h(s.id)-h(\hat{v})|$ is minimized.
	Then \[|h(s.id)-h(\hat{v})| \in \left[\frac{1}{n^2}, \frac{\log n}{n}\right] \text{ w.h.p.}\]
\end{lemma}

\begin{lemma} \label{lemma:log_n}
	Given $n$ clients $v_1,\ldots,v_n$ and a server $s$.
	Once at least $n/\log n$ arbitrary clients have pinged the server it holds \[\Delta = - \log|h(s.id)-h(\hat{v})| \in \Theta(\log n) \text{ w.h.p.}\]
\end{lemma}

\begin{proof}
	Using \Cref{lemma:distance_approximation}, we know that \[|h(s.id)-h(\hat{v})| \in \left[\frac{\log^2n}{n^2}, \frac{\log^2 n - \log n \cdot \log \log n}{n}\right] \text{ w.h.p.}\] after $n/\log n$ arbitrary clients have pinged the server.
	Computing $-\log|h(s.id)-h(\hat{v})|$ results in a value within the interval \[\left[\log n - \log(\log^2n-\log n\cdot \log \log n), 2\log n - \log(\log^2 n)\right] \text{ w.h.p.}\] which is within $\Theta(\log n)$.
\end{proof}

Note that it is easy to verify that once all $n$ clients have pinged the server at least once, then $-\log|h(s.id)-h(\hat{v})|$ is within $\Theta(\log n)$ as well.

\subsection{Strategy 2: Maintaining a Table} \label{app:approx_log_n:2}
In the following we introduce a more involved strategy that takes ideas from Strategy $1$, but is able to handle initially corrupted entries at the server.

\paragraph{Description}
Assume the server stores a value $N$ which it thinks is equal to $n$.
The server maintains a table (see \Cref{table:log_n}) with $\log \log N$ columns denoted by $c_0,\ldots,c_{\log \log N - 1}$, where column $c_i$ represents the value $\sqrt[2^i]{N}$.
Furthermore there is a timestamp $t_i \in \mathbb{N}_0$ for all columns $c_i$.

The table is maintained as follows by the server: Same as for the previous strategy, we map the identifiers of the server and the clients to the interval $[0,1)$ via a uniform hash function $h: \mathbb{N} \rightarrow [0,1)$.
Whenever a client $v$ with $|h(s.id) - h(v.id)| \leq \frac{1}{c_i}$ successfully pings the server, the server (re-)sets all timestamps $t_i,\ldots,t_{\log \log N -1}$ to $0$.
Aside from this, each timestamp $t_i$ gets incremented by one in each round.
Once the entry $t_i$ for column $c_i$ gets larger than $\mathfrak{c} \cdot c_i \cdot O(polylog(c_i))$ for an appropriately chosen constant $\mathfrak{c} \geq 1$, all columns $c_0,\ldots,c_i$ are deleted from the table and the value $N$ is set to the column $c_{i+1}$.
On the other side, once a client $v$ pings for which $|h(s.id) - h(v.id)| \leq \frac{1}{c_0^2}$ holds we update the table by adding that many columns to the left until $\frac{1}{c_0^2} < |h(s.id) - h(v.id)| \leq \frac{1}{c_0}$ holds.
The server always sets $\Delta = \Theta(\log c_0)$ to approximate $\Theta(\log n)$.

\paragraph{Analysis}
We now show that eventually $c_0 \in O(n)$ and thus $\log c_0 \in \Theta(\log n)$ holds.
For this we use the fact that after at most $O(p_{\mathit{min}}^{-1} + n \cdot \log^2 n)$ rounds at least $\Theta(n/\log n)$ clients have successfully pinged the server at least once.
This is shown in the proof of \Cref{lemma:convergence_time_0} in \Cref{app:sec:n_pings}.

We divide the set of clients into \emph{visible} and \emph{invisible} clients.
A client $v$ is visible if its probability is within $\Omega\left(\frac{P}{n \cdot polylog(n)}\right)$.

Fix an arbitrary state $s$ and let $n_{vis}$ be the number of visible clients in $s$.
We show the following lemma:

\begin{lemma} \label{app:lemma:vis}
	After $O(n^2 \cdot polylog(n))$ rounds it holds $\Delta \leq \Theta(\log n)$ at the server.
\end{lemma} 

\begin{proof}
	Consider the $i^{th}$ column such that $c_i = \Theta(n^2)$.
	By \Cref{lemma:distance_approximation} there is no visible client $v$ such that $|h(s.id)-h(v.id)| \leq \frac{1}{n^2}$ and thus after $c_i \log c_i = \Theta(n^2 \log n)$ rounds the columns $c_0,\ldots,c_i$ are deleted from the table. 
	This means that the server can only set $\Delta$ to some value that is less or equal to $\Theta(\log n)$.
\end{proof}

\Cref{app:lemma:vis} implies that after at most $O(n^2 \cdot polylog(n))$ rounds, the server got rid of corrupted values from initial states.
In the following we therefore can assume that $\Delta \leq \Theta(\log n)$ holds.

\begin{lemma} \label{app:lemma:all_visible}
	Let $\Delta \leq \Theta(\log n)$.
	After $O(p_{\mathit{min}}^{-1} + n \cdot \log^2 n)$ rounds $\Theta(n/\log n)$ clients are visible w.h.p.
\end{lemma}

\begin{proof}
	Due to \Cref{lemma:convergence_time_0} we know that $\Theta(n)$ clients have pinged the server successfully at least once after $T := O(p_{\mathit{min}}^{-1} + n \cdot \log^2 n)$ w.h.p.
	Denote the set of these clients by $V'$.
	For each client $v \in V'$ consider the very last of its successful pings after $T$ rounds.
	Let this ping happen in round $t \leq T$.
	Let the random variable $X_t(v) \in \{0,1\}$ indicate if $v$ is visible ($X_t(v)=1$) or if $v$ is invisible ($X_t(v)=0$) before round $t$.
	If $X_t(v) = 1$ then $X_{t+1}(v) = 1$ either, because $v$ can only be reduced once in a worst-case by a constant factor $(1+\frac{1}{\sigma})$.
	Now assume that $X_t(v)=0$.
	If $P \geq L$ then, with at least constant probability, $v$ is paired with a visible node $w$ and the average-rule is applied.
	This leads to $X_{t+1}(v)=1$ and $X_{t+1}(w)=1$ ($w$ is only reduced by a constant factor less than $2$).
	In case $P < L$ it holds that with probability at least $\frac{1}{\Delta} \geq \frac{1}{\log n}$ $v$ gets its probability increased to $\hat{p}$, so $X_{t+1}(v) = 1$.
	
	Now, summing up over all $X_{t+1}(v)$, we get that on expectation $$X=\sum_{v \in V'} X_{t+1}(v) \geq \frac{1}{\log n} \cdot \Theta(n) = \Theta(n/\log n)$$ clients are visible after $T$ rounds.
	Applying a Chernoff bound to this value yields the lemma, so for any $0 \leq \delta \leq 1$ it holds \[\Pr[X \leq (1-\delta)E[X]] \overset{\Cref{lemma:chernoff}(iii)}{\leq} \exp\left(\frac{-\delta^2 \cdot \Theta(n/\log n)}{2}\right) \leq \exp(-cn)\] for an appropriate constant $c$.
\end{proof}

We can now show that our strategy converges:

\begin{lemma} \label{app:lemma:approx_convergence}
	After $O(p_{\mathit{min}}^{-1} + n^2 \cdot polylog(n))$ rounds $\Delta \in \Theta(\log n)$ and for all $t_i$ it holds $t_i = 0$ w.h.p.
\end{lemma}

\begin{proof}
	Combining \Cref{app:lemma:vis} and \Cref{app:lemma:all_visible}, we know that after $O(p_{\mathit{min}}^{-1} + n^2 \cdot polylog(n))$ rounds all superfluous columns have been deleted and $\Theta(n/\log n)$ clients are visible.
	Thus an application of \Cref{lemma:log_n} implies that there exists at least one visible client whose ping triggers the server to include the column $c=\Theta(n)$ to the table.
	Also all $t_i$ will be reset to $0$ by this ping, so all in all the table becomes good and the server sets $\Delta=\Theta(\log n)$.
\end{proof}

It remains to show that the approximation remains stable once it reached a correct value.
We first show that visible clients remain visible:

\begin{lemma}
	Once we reached a state where it holds $\Delta \in \Theta(\log n)$, $t_i = 0$ for all $t_i$ and $\Theta(n/\log n)$ clients are visible, then visible clients remain visible throughout the algorithm w.h.p.
\end{lemma}

\begin{proof}
	This can be easily shown by adapting the behavior of the server as follows: The server is only allowed to set probabilities to a value greater or equal than $\frac{1}{c_0 \cdot polylog(c_0)}$.
	So once the approximation is correct, this corresponds to a probability of $\Omega\left(\frac{1}{n \cdot polylog(n)}\right)$.
	Note that the above described adaptation does not hurt the convergence of the approximation: Probabilities that are too small are increased to some (potential large value initially), but this is only beneficial in order to get out of initial states where $P$ is too low.
\end{proof}

\begin{lemma} \label{app:lemma:approx_closure}
	Once we reached a state where it holds $\Delta \in \Theta(\log n)$, $t_i = 0$ for all $t_i$ and at least $\Theta(n/\log n)$ clients are visible, $\Delta$ remains in $\Theta(\log n)$ w.h.p.
\end{lemma}

\begin{proof}	
	The only time when the table gets extended is when more clients become visible.
	However, applying \Cref{lemma:log_n} to $n$ clients instead of $\Theta(n/\log n)$ clients still yields a value $\Delta \in \Theta(\log n)$, as $c_0$ cannot become larger than $\Theta(n^2)$, so $\Delta$ remains in $\Theta(\log n)$ in this case.
	
	It remains to show that w.h.p. no further column gets deleted from the table.
	Consider the visible client $v$ with $|h(s.id)-h(v.id)| \leq \frac{1}{c_0}$.
	Since $v$ is visible it holds that, w.h.p., $v$ pings the server within at most $O\left(\frac{n \cdot polylog(n)}{P}\right)$ rounds (\Cref{lemma:whp_ping}).
	Since by \Cref{lemma:runtime:protocol:0.1} it takes only $O(n \log^2 n)$ rounds to let $P$ get to a value that is within $O(\log n)$, it follows that after at most $O(n \cdot polylog(n))$ rounds $v$ will ping the server and be successful with probability at least $1/\log n$.
	Thus, w.h.p., it holds that $v$ pings the server successfully after at most $O(n \cdot polylog(n) \cdot \log n) = O(n \cdot polylog(n))$.
	For $\mathfrak{c} \geq 1$ chosen appropriately, this time is less than $\frac{1}{2} \cdot \mathfrak{c} \cdot c_0 \cdot O(polylog(n))$, so the column $c_0$ will reset its entry $t_0$ in time.
	Therefore we know that once we reached a state $s_0$ where it holds $\Delta \in \Theta(\log n)$, $t_i = 0$ for all $t_i$ we will reach a subsequent state $s_1> s_0$ with the same properties and for all states $s$ between $s_0$ and $s_1$ it holds $\Delta \in \Theta(\log n)$ as well.
\end{proof}

Combining \Cref{app:lemma:approx_convergence} with \Cref{app:lemma:approx_closure} yields \Cref{theorem:log_n:approx}.

Finally we show that the table at the server does require only $\Theta(\log N)$ bits of memory, i.e., we satisfy the memory constraint at the server in legitimate states.

\begin{lemma} \label{app:lemma:table_memory}
	Let $N = O(n)$.
	Then \Cref{table:log_n} can be stored using $\Theta(\log n)$ bits.
\end{lemma}

\begin{proof}
	Assume w.l.o.g. that $N = 2^{2^k}$ for some $k \in \mathbb{N}$.
	Consider a column $c_i$ of \Cref{table:log_n}.
	Then $c_i = 2^{2^{k-i}}$.
	By definition of the protocol the value $t_i$ cannot get higher than $\Theta(c_i  \cdot polylog(c_i))$.
	Thus we need $\Theta(\log(c_i \cdot polylog(c_i))) = \Theta(\log c_i) = \Theta(2^{k-i})$ many bits to store $t_i$.
	Summing up over all $\log \log N$ columns we get \[\sum_{j=0}^{\log \log N - 1} \Theta(2^j) = \Theta(\log N) = \Theta(\log n)\] for the overall number of bits required to store \Cref{table:log_n}.
\end{proof}

\subsection{Discussion}
Both of the above proposed solutions come with advantages and disadvantages.

Strategy $1$ has the advantage that it is quite simple and easy to implement.
However, if we assume the existence of corrupted client identifiers, then Strategy $1$ is not able to deal effectively with certain initial states: One may ask what happens if in an initial state, the server stores a corrupted identifier in $\hat{v}$, i.e., the server may think that there exists a client with identifier $\hat{v}$, although that is not the case.
This may lead to a wrong approximation of $\log n$, i.e., the server may store a much higher value than $\Theta(\log n)$ in $\Delta$.
Although this does not hurt the overall correctness of the protocol, it may negatively influence the overall convergence time depending on the initial value of $\hat{v}$.
Still, even if there do not exist corrupted identifiers, it is not clear how long this approach takes to converge, because we need to wait for $n/\log n$ clients to successfully ping.
Note that we are not allowed to apply \Cref{lemma:convergence_time_0} here, as we cannot apply a similar argumentation for its use as we have done for Strategy $2$.

Strategy $2$ solves the above mentioned problem, as it comes with a quite effective way to deal with arbitrary initial states.
The time it takes to remove corrupted entries in the table is $O(n^2 \cdot polylog(n))$.
Notice that for cases in which the table is too small initially, the additional time of $O(n^2 \cdot polylog(n))$ rounds is not needed (cf. the proof of \Cref{app:lemma:all_visible}).
This makes Strategy $2$ a truly self-stabilizing solution for approximating $\log n$.

\section{Impact of Negative Feedback}\label{app:negative_feedback}
In this section we observe the impact that the negative feedback has on the probabilities.

\begin{proof}[Proof of \Cref{lemma:negative_feedback}]
For simplicity we will denote $p_{\mathit{min}} := p_0$ and $p_i := p_0 (1+\frac{1}{\sigma})^{-i}$ for short throughout this proof.
Our goal is to show that no node is assigned a probability $\leq p_{i^*}$ with $i^* := \log{\log{n}}+c_1$ with probability $1-\frac{1}{n^{c_1}}$.
In order for a node to be assigned any probability $p_{i+1}$ with $i>0$,
there are two possibilities:
\begin{enumerate}
    \item Two nodes of probability $p_i$ (and no node of higher probability) must ping the server when a ball is thrown.
    \item A node of probability $p_i$ is paired and balanced with a node of probability $p_{i+1}$ (or lower). 
\end{enumerate}
Note that for every time the second event occurs, the first must have occurred at least once. 
Thus, in the remainder of this proof, we will not consider the balancing step and assume a ball thrown every round.

We will now use the \emph{layered induction} approach popularized in
many other papers dealing with the $c$-choice algorithms.
However, we need to account for three deviations between our model and the standard model:
\begin{enumerate}
    \item First, we throw more than $O(n)$ balls.
    \item Second, our bins are not picked with uniform probability.
    \item Third, the probabilities of the dynamically changing.
\end{enumerate}
As is turns out, the properties nicely \emph{cancel out} each other as \emph{full} are less likely to receive more balls.

In the remainder let $\eta_0 \leq n$ be maximal number of nodes with probability between $p_0$ and $p_1$ in any round round $t \in [0,T]$. 
Likewise, let $\eta_i$ the maximal number of all nodes with probability in the area $[p_i,p_{i+1}]$.
We claim that the following holds:
\begin{claim}
Let $\beta_0<1$ be a constant.  
Then it holds as long as $\eta_i \in \Omega(\log n)$
$$
    \eta_i \leq \beta_0^{2^i} n
$$
w.h.p.
\end{claim}
\begin{proof}
We proof the statement by induction. 
Therefore, we first show that there is a linear fraction of nodes that never see the impact of any of the balls.
Then, we show that given only a (small) subset of nodes get $i$ balls, an even smaller subset sees $i+1$ balls.
All of the above statements hold w.h.p.
\begin{itemize}
    \item[(IB)] $\eta_0 \leq \frac{n}{c}$ for some $c>0$
        
	Recall that we wish to bound the total number of nodes with probability $p_v(t) < p$ in any round.
	Let $v \in V$ be any node whose probability is reduced below $p$.
	This implies that in the round \emph{before} $v$ was reduced below $p$, 
	it had a probability of at most $2p$ as any ball reduces $p_v(t)$ by a constant factor.
	Thus, we the upper bound the probability that $v$ gets reduced if we simply assume that 
	$v$ always has a probability of $2p$ if it is not below $p$.
	
	For every node $v$ let $X^T_v$ be binary random variable that denotes if $v$ is reduced below $p$ within $T$ rounds.
    	By the observation above, we see that it holds for $T := p^{-1}$:
    	$$
    		\Pr[X_v^T] \leq 1-(1-2p)^T \leq 1-\frac{1}{e^2} 
    	$$
    	Thus, the expected fraction of nodes that are \emph{never} decreased below $p$ is at least $\frac{1}{e^2}$
    	Using a standard argument and the Chernoff bound, 
    	we can easily show that (at most) $\left(1-\frac{e^{-c_1}}{2}\right)n$ are reduced w.h.p.

    \item[(IS)] $\eta_i \leq \beta^{(2^i)}\eta_0 \Longrightarrow \eta_{i+1} \leq \beta^{(2^{i+1})}\eta_0$
    
    Now we consider the induction step.
    Therefore, we condition on the event that everything went according to the hypothesis so far.
    Thus, over the whole execution of the algorithm there is no round where more than $\eta_i \leq \beta^{2^i}\eta_0$ nodes have their probabilities in $[p_i,p_{i+1}]$.
	We are now interested in the fraction of nodes that are reduced below $p_{i+1}$.
	We apply the following three simplifications to upper bound this number.
    \begin{enumerate}
    \item First, we assume that in \emph{all} rounds \emph{all} $\eta_{i}$ nodes have prob. $p_i$. 
    This only makes them more likely to send.
    \item Second, we only consider the case that exactly two nodes with prob. $p_i$ ping and ignore nodes of higher probabilities. 
    Surely, this makes the event that the ball picks such a node only more likely.
    \item Last, consider that the number of balls thrown into bin with probability $p_{i+1}$ certainly upper bounds the number of these bins. 
    \end{enumerate}
    Therefore, we consider the following random experiment:
    For every of the $p^{-1}$ balls thrown we consider the event that two of the $\eta_i$ nodes with probability $p_i$ send.
   	By the union bound, the probability that a given ball $j \in 0, \dots, p^{-1}$ is thrown into a bin with probability $p_i$ is bounded by:
    	\begin{equation} \label{eqn:union_bound}
    		\Pr[X_j = 1] \leq \binom{\eta_i}{2}p_i^2 \leq c_1 \left(\frac{\eta_i}{p_i}\right)^2
    	\end{equation}
    	Here, $c_1$ is a constant that results from the application of Stirling's formula.
    	Now consider the expected number of balls:
    	\begin{align*}
    	\mathbb{E}[\eta_{i+1}] \leq \mathbb{E}[\sum_{j=1}^{p^{-1}} X_j] &= \sum_{j=1}^{p^{-1}} \mathbb{E}[X_j] & \textit{(Linearity of expectation)}\\
    	&\leq p^{-1} c_1 \left(\frac{\eta_i}{p_i}\right)^2 & \textit{(Equation \ref{eqn:union_bound})}\\
    	&\leq p^{-1} c_1 \left(\frac{\beta_0^{(2^i)}\eta_0}{p_i}\right)^2 & \textit{(Induction Hypothesis)}\\
    	&= c_1 p \left(\frac{\beta_0^{(2^i)}}{(1+\frac{1}{\sigma})}\right)^2 \eta_0^2 & \textit{(Using } p_i = p\left(\frac{1}{1+\frac{1}{\sigma}}\right)^i)\\
    	&\leq \left(\frac{\beta_0^{(2^i)}}{(1+\frac{1}{\sigma})}\right)^2 \eta_0 & \textit{(Since $p \in o(\frac{1}{n}) \Rightarrow p \leq \frac{1}{c_1n}$)}\\
	&\leq \frac{\beta_0^{(2^{i+1})}}{(1+\frac{1}{\sigma})^2} \eta_0    	
    	\end{align*}
    Since all balls are thrown independently of one another,
    we can apply the Chernoff bound and get:
    $$
        \Pr[\eta_{i+1} \geq (1 + \delta)\mathbb{E}[\eta_{i+1}]] := \Pr[\eta_{i+1} \geq (1 + \delta)\beta \eta_i] \leq exp\left[\delta\mathbb{E}[\eta_{i+1}] \right]
    $$
    If we choose $\delta := \frac{1}{\sigma}$ it holds 
    $$
    \Pr\left[\eta_{i+1} \geq \frac{\beta_0^{(2^{i+1})}}{(1+\frac{1}{\sigma})} \eta_0\right] \leq  exp\left[\frac{1}{\sigma}\mathbb{E}[\eta_{i+1}]\right]
    $$    
    Thus, as long as $\eta_{i+1} \in \Omega(\log n)$ this holds w.h.p.
\end{itemize}{}
\end{proof}{}
We can now finalize our proof. 
From the induction we can conclude that $\eta_{\log{\log{n}}} \in O(\log{n})$ w.h.p. as $\frac{1}{c}^{2^{\log\log(n)}}n \in O(1)$.
In order words, there are only $O(\log{n})$ nodes (that started with probability at least $p$) that received $\log{\log{n}}$ balls.
Now we observe the probability that any of these nodes obtains another $c_3$ balls.
\begin{align*}
    \Pr[c_3] &= \binom{p}{c_3}\left(\log(n)p\right)^{2c_3}\\
    &\leq \left(\frac{e}{c_3 \cdot p}\right)^{c_3} \left(\log(n)p\right)^{2c_1}\\
    &\in O(p^{c'_3}) \in o(n^{-c'_3})
\end{align*}
Thus no node with $p_{\log{\log{n}}}$ gets more than a constant number of additional balls. 
This proves the lemma.
\end{proof}

\section{Time until \texorpdfstring{$\Theta(n/\log n)$}{Theta(n/log n)} Clients Pinged Successfully} \label{app:sec:n_pings}
In this section we upper bound the time it takes until $\Theta(n/\log n)$ clients have pinged the server successfully at least once, when starting from any arbitrary initial state.
More specifically we show the following:

\begin{lemma} \label{lemma:convergence_time_0}
	Consider any arbitrary initial state.
	After $O(p_{\mathit{min}}^{-1} + n \cdot \log^2 n)$ rounds, $\Theta(n/\log n)$ clients have successfully pinged the server at least once w.h.p.
\end{lemma}

The following lemmas are required for the proofs that follow:

\begin{lemma} \label{lemma:app:ping_amount}
	Consider the set $V' \subset V$ such that $V' = \{v \in V\ |\ p(v) \geq \frac{P}{2n}\}$. 
	While $P > \log n$, at least one client $v \in V'$ pings the server in each round w.h.p.
\end{lemma}

\begin{proof}
	Consider the ordered sequence $p(v_1) \leq p(v_2) \leq \ldots \leq p(v_n)$ of all client probabilities.
	For this sequence we choose the index $k \in \{1,\ldots,n\}$ such that $\sum_{i=1}^{k-1} p(v_i) = \frac{P}{2}$.
	We claim that $p(v_k) \geq \frac{P}{2n}$: Assume to the contrary that $p(v_k) < \frac{P}{2n}$.
	Then it would follow that \[\frac{P}{2} = \sum_{i=1}^{k-1} p(v_i) < \sum_{i=1}^{k-1} \frac{P}{2n} = (k-1) \cdot \frac{P}{2n}.\]
	This implies that $k \geq n+1$, which is a contradiction, so the claim holds.

	Now consider the set $V' = \{v \in V\ |\ p(v) \geq \frac{P}{2n}\}$.
	Using the claim from above, we get \[\sum_{v\in V'}p(v) \geq \frac{P}{2}.\]
	
	We are now ready to show the lemma: For $P > \log n$ it holds that \[\sum_{v\in V'}p(v) > \frac{\log n}{2} \in \Theta(\log n).\]
	We now claim that at least one client out of $V'$ pings the server in each round w.h.p.: Let $X_j = 1$ if and only if at least one client $v \in V'$ pings the server in round $j$.
	Then \[\Pr[X_j = 1] \geq 1 - \left(1-\frac{P}{2n}\right)^n = 1 - e^{-P/2} \geq 1 - e^{-\frac{\log n}{2}} \geq 1-n^{-c},\] for any $j$ and a constant $c>0$, so the lemma holds.
\end{proof}

For the next lemma denote by $\mathcal E_{\sigma}^t \in \{ V' \subset V \big| \, |V'| \leq \sigma \}$ the set of at most $\sigma$ clients that have successfully pinged the server in round $t$. 

\begin{lemma} \label{lemma:pre:ping_amount2}
	Let $P \leq n/2$.
	Consider a client $v \in V$ with probability $p(v) \in (0,\hat{p}]$.
	Then for any round $t$ it holds $\Pr[v \in \mathcal E_{\sigma}^t] \geq \frac{p(v)}{4P}$.
\end{lemma}

\begin{proof}
	The following holds:
	\begin{eqnarray*}
		\Pr[v \in \mathcal E_{\sigma}^t] & = & \sum_{k=1}^n \Pr[v \in \mathcal E_{\sigma}^t\ |\ k \text{ clients ping}] \cdot \Pr[k \text{ clients ping}]\\
	\end{eqnarray*}
This follows directly from the law of total probability.
Let us now observe the event that $v$ successfully pings the server given that $k$ clients ping the server. In order to successfully ping, the following two events need to happen:
\begin{enumerate}
\item $v$ must send a message to the server in round $t$. This happens with probability $p_v$.
\item $v$'s message must not be dropped by the server. This happens with probability $\max\{1, \frac{\sigma}{k}\}$ given that at most $\sigma$ succeed and $k$ nodes tried to ping.
\end{enumerate}
Combining the these two facts yields that $v$ successfully pings with probability at least $\frac{1}{k}$.

Thus, we can further simplify our equation:
\begin{eqnarray*}
		\Pr[v \in \mathcal E_{\sigma}^t] & = & \sum_{k=1}^n \frac{p_v}{k} \cdot \Pr[k \text{ clients ping}] \\
		& \geq & \sum_{k=1}^{2P} \frac{p(v)}{k} \cdot \Pr[k \text{ clients ping}] \\
		& \geq & \sum_{k=1}^{2P} \frac{p(v)}{2P} \cdot \Pr[k \text{ clients ping}] \\
		& = & \frac{p(v)}{2P} \cdot \sum_{k=1}^{2P} \Pr[k \text{ clients ping}] \\
		& = & \frac{p(v)}{2P} \cdot (\Pr[k \leq 2P \text{ clients ping}\\
		& = & \frac{p(v)}{2P} \cdot (1- \Pr[k > 2P \text{ clients ping}]) \\
\end{eqnarray*}
Now observe $\Pr[k > 2P \text{ clients ping}]$. 
Let $X^t := \sum X_v^t$ be a random variable that denotes the number of clients that ping in round $t$, i.e., it holds $X_v^t=1$ if and only if client $v \in V$ pings in round $t$.
Following the linearity of expectation, it holds that the expected number of clients that ping in a given round is $P$. Thus, we see that:
\[
\Pr[k > 2P \text{ clients ping}] = \Pr[X > 2\mathbb{E}[X]] 
\]
By \Cref{lemma:markov} it holds
\[
	\Pr[X > 2\mathbb{E}[X]] \leq \frac{1}{2}
\]			
Therefore, it holds
\begin{eqnarray*}		
		\Pr[v \in \mathcal E_{\sigma}^t] &{\geq} & \frac{p(v)}{2P} \cdot (1- \Pr[k \geq 2 \cdot \mathbb{E}[\mathcal E_{\sigma}^t] \text{ clients ping}] ) \\
		& \geq & \frac{p(v)}{2P} \cdot \frac{1}{2} \\
		& = & \frac{p(v)}{4P} \\
	\end{eqnarray*}
This was to be shown.
\end{proof}

Intuitively \Cref{lemma:pre:ping_amount2} implies that for $P \in \Theta(1)$ the probability for a client to successfully ping is not significantly lower than the probability to just ping.
Therefore we can compute the amount of rounds until $\Theta(n)$ clients have pinged successfully at least once in case $P \in \Theta(1)$:

\begin{lemma} \label{lemma:runtime:protocol:0.4}
	Let $P \in \Theta(1)$.
	After $O(p_{\mathit{min}}^{-1})$ rounds, $\Theta(n)$ successfully pinged the server at least once w.h.p.
\end{lemma}

\begin{proof}
	Consider any client $v \in V$ and a time frame of $T=4P \cdot p_{\mathit{min}}^{-1}$ rounds.
	Define the variable $X_v = 1$ if $v$ has successfully pinged the server within these $T$ rounds (otherwise $X_j = 0$).
	By \Cref{lemma:pre:ping_amount2} it holds for a single round $j$ that $v$ successfully pinged the server in round $j$ with probability at least $\frac{p_{\mathit{min}}}{4P}$, so $E[X_v] = T \cdot \frac{p_{\mathit{min}}}{4P} = 1$.
	Let $X = \sum_{v \in V} X_v$.
	Then $\mathbb{E}[X] = n$.
	Using \Cref{lemma:chernoff}(iii) we can compute the probability that less than a constant fraction $(1-\delta)$ of $n$ clients ping the server successfully at least once within $O(p_{\mathit{min}}^{-1})$ rounds:
	\begin{eqnarray*}
		\Pr[X=0] & = & \Pr[X \leq (1-\delta)\cdot n]\\
		& \leq & \exp\left(\frac{-(\delta^2 \cdot n)}{2}\right) \\
		& \leq & n^{-c}
	\end{eqnarray*}
	for some constant $c > 0$, so we know that $\Theta(n)$ clients ping the server successfully within the time frame $T \in O(p_{\mathit{min}}^{-1})$ w.h.p.
\end{proof}

We still have to worry about initial states, where $P$ is not constant.
For this we first show an upper bound on the number of rounds it takes to reduce $P$ to a constant when starting from an initial state with $\Delta \leq \Theta(\log n)$ (we consider the time needed to achieve this in \Cref{app:approx_log_n:2}).

\begin{lemma} \label{lemma:runtime:protocol:0.1}
	Let initially $P > O(\log n)$ and assume $\Delta \leq \Theta(\log n)$.
	Then $P \leq \log n$ after $O(n \log^2 n)$ rounds w.h.p.
\end{lemma}

\begin{proof}
	Since $P > O(\log n)$ it is easy to see that the server approximates $P \succ R$ with at least constant probability (regardless of the initial size of $\Delta$) and thus starts reducing probabilities thereafter.
	
	Now fix $P > O(\log n)$.
	Due to \Cref{lemma:app:ping_amount} we know that that at least one client out of $V' = \{v \in V\ |\ p(v) \geq \frac{P}{2n}\}$ pings the server in each round w.h.p.

	Once a single client $v \in V'$ pings the server in a round where the server approximates $P$, $P$ is reduced by at least \[\frac{P}{2n} - \frac{P}{2n(1+1/\sigma)} = \frac{P}{2n (1 + \sigma)},\] i.e., in the next round $P$ is equal to $P-\frac{P}{2n(1+\sigma)} = (1-\frac{1}{2n(1+\sigma)}) \cdot P$.
	Thus, after $T$ reductions we have that $P$ is equal to $$\left(1-\frac{1}{2n(1+\sigma)}\right)^T \cdot P.$$
	Setting $T = 2n(1+\sigma)\log n = O(n \log n)$ yields $P \in \Theta(1)$, so $O(n \log n)$ reductions suffice.
	As one reduction occurs every $\Theta(\log n)$ rounds, these $O(n \log n)$ reduction can be achieved in $O(n \log^2 n)$ rounds w.h.p.
\end{proof}

Next, we consider the case that $P \leq \log n$ initially.
We will show that after $O(p_{\mathit{min}}^{-1})$ rounds it holds w.h.p. that at least $\frac{n}{2 \log n}$ different clients have successfully pinged the server.
Formally, we show the following.
	
\begin{lemma}\label{lemma:runtime:protocol:0.2}
Assume $P \leq \log n$ and let $X:=(X^T_v)_{v \in V}$ be a set of random variables such that each $X_v^T \in \{0,1\}$ denotes the event that $v$ successfully sends at least once within $T$ rounds. Then, for $T \in O(p_{\mathit{min}}^{-1})$ it holds 
\[
	\Pr\left[\sum_{v \in V} X^T_v  \geq \frac{n}{\log n}\right] \geq 1-\frac{1}{n^k}.
\]
\end{lemma}

The proof has two steps: First, we show that the expected value of $\sum_{v \in V} X^T_v$ is at least $n/\log n$.
Using this lower bound, we apply Chernoff Bounds to these variables.

We begin with the calculation of the expected value. 
We first observe the probability for the event that a node successfully pings the server and show the following:

\begin{lemma} \label{lemma:pv_lower_bound}
Let $(p_1, \dots, p_n) \in [0,1)^n$ be the initial sending probabilities for all $v \in V$ and let $P \leq \log n$.
Let $A_v \in \mathbb{N}$ denote in which round $v$ pings successfully for the first time. Then it holds w.h.p.
\[
	\Pr[A_v = t] \geq \frac{p_{\mathit{min}}}{4\log n}
\]
\end{lemma}

\begin{proof}
We can prove this statement by a simple induction over all rounds.
Therefore let $v$ be any node that has not pinged until round $t$
and let $p_t(v)$ be its sending probability in round $t$.

The induction beginning for $t=0$ follows from \Cref{lemma:pre:ping_amount2} and the fact that $P \leq \log n$. 
Using that lemma we see that the probability to send is bigger than $\frac{p_t(v)}{4\log n}$. Since $p_0(v) \geq p_{\mathit{min}}$ per definition, the lemma follows.

For the induction step recall that a message by the server is the only action that causes a node to reduce its probability.
However, for that it is necessary that the node pinged successfully at least once. Otherwise, the server will never send the node a new probability.
Since $v$ did not successfully send until round $t$, it still holds $p_t(v) = p_0(v)$. 
Since $P$ does not grow bigger than $O(\log n)$ w.h.p., the lemma follows.
\end{proof}

Using this probability we can calculate the \emph{expected} number of nodes that successfully send a message at least once within $T$ rounds. That is

\begin{lemma} \label{lemma:rv_lower_bound}
Let $T \geq 4p_{\mathit{min}}^{-1}$. For $X^T := \left(X^T_v\right)_{v \in V}$ it holds.
\[
	\mathbb{E}\left[\sum_{v \in V} X_v^T\right] \geq \frac{n}{\log n}
\]
\end{lemma}

\begin{proof}
Due to \Cref{lemma:pv_lower_bound} the probability that a node $v \in V$ pings successfully at least once over course of $T$ rounds is at least $\frac{p_{\mathit{min}}}{4 \log n}$.
Therefore the expected number of nodes that successfully pinged the server at least once within $T$ rounds is $$\mathbb{E}\left[\sum_{v \in V} X_v^T\right] = T \cdot \frac{p_{\mathit{min}}}{4 \log n} \geq \frac{n}{\log n}.$$
\end{proof}

We now wish to apply the Chernoff bound to this result to show that it also holds with high probability.
Note that the events that a node successfully pings in a given round is dependent of the events in the previous rounds. 
Thus, the corresponding events are not independent and we cannot trivially apply the Chernoff bound. 
However, since we obtained a lower bound for $\mathbb{E}\left[\sum_{v \in V} X_v^T\right]$ in \Cref{lemma:rv_lower_bound} we are still allowed to use Chernoff bounds, as it has been shown in~\cite{hniid=510}.

Thus, a simple application of the Chernoff bound concludes the proof for \Cref{lemma:runtime:protocol:0.2}. 

\begin{proof}[Proof of \Cref{lemma:runtime:protocol:0.2}]
Let $\delta > 0$ be a constant.
Using \Cref{lemma:chernoff}(iii) we get
	\begin{eqnarray*}
		\Pr\left[\sum_{v \in V} X_v^T  < (1-\delta) \frac{n}{\log n}\right] 
		& \leq & \exp\left(\frac{-(\delta^2 \cdot n)}{2 \cdot \log n}\right) \\
		& \leq & n^{-c}
	\end{eqnarray*} 
	for a constant $c > 0$.
\end{proof}

We obtain \Cref{lemma:convergence_time_0} by combining~\Cref{lemma:runtime:protocol:0.4},~\Cref{lemma:runtime:protocol:0.1}, and~\Cref{lemma:runtime:protocol:0.2}.

\section{Correctness Analysis} \label{app:correctness_analysis}
We prove the following theorem in this section:

\begin{theorem}\label{theorem:self_stabilization}
	The protocol is self-stabilizing with regard to $\mathfrak{busy}$ and $\mathfrak{safe}$ states.
\end{theorem}

First of all note that in case there are corrupted messages in the system initially, these will be processed within one (synchronous) round by their receiver.
After the round is over no more corrupted messages exist leaving only corrupted information in variables, which our protocol is able to deal with.

We show that the $P$ converges to a value within $(L+\varepsilon,R-\varepsilon)$, hence is $\mathfrak{busy}$
In order to do so, we prove the correctness of the server's approximation for $P$.

\begin{lemma} \label{lemma:approx_correctness}
	Let $P$ be fixed for the last $\Delta$ rounds.
	\begin{itemize}
		\item[(i)] If $P > L + \varepsilon$ then $P \succ L$ at the server w.h.p.
		\item[(ii)] If $P < L - \varepsilon$ then $P \prec L$ at the server w.h.p.
		\item[(iii)] If $P > R + \varepsilon$ then $P \succ R$ at the server w.h.p.
		\item[(iv)] If $P < R - \varepsilon$ then $P \prec R$ at the server w.h.p.
	\end{itemize}
\end{lemma}

\begin{proof}
	We only proof the first statement, since the others work analogously.
	Let $P > L + \varepsilon$.
	We choose \[\Delta = \left\lceil \frac{2 \sigma \mathfrak{c} \log n}{(L+\varepsilon) \cdot \left(1-\frac{L}{L+\varepsilon} \right)^2} \right\rceil\] for a constant $\mathfrak{c} \geq 1$.
	Note that $\Delta \in \Theta(\log n)$, as $\sigma, L$ and $\varepsilon$ are constants.
	Consider the random variables $X_1,\ldots,X_{\Delta} \in \{0,\ldots,\sigma \}$ where $X_j$ denotes the number of clients that successfully pinged in round $j$.
	Then $\mathbb{E}[X_j] > L + \varepsilon$ for each $j \in \{1,\ldots,\Delta\}$.
	Define $X = \sum_{i=1}^{\Delta} X_i$.
	Then $\mathbb{E}[X] = \sum_{i=1}^{\Delta} \mathbb{E}[X_i] > \Delta (L+\varepsilon)$.
	Note that since we got a lower bound for $\mathbb{E}[X]$, we can still apply (generalized)  Chernoff bounds, even though the random variables $X_1,\ldots,X_{\Delta}$ are not independent.
	This has been shown in~\cite{hniid=510}.
	We compute the probability that the server approximates $P\succ L$ now: By definition of our protocol the server approximates $P\succ L$ if $X/\Delta > L$.
	Thus we get:
	\begin{eqnarray*}
		\Pr[X/\Delta > L] & = & \Pr[X > \Delta \cdot L]\\
		& = & \Pr\left[X > \left(\frac{L}{L+\varepsilon}\right) \Delta \cdot (L+\varepsilon)\right] \\
		& = & 1 - \Pr\left[X \leq \left(\frac{L}{L+\varepsilon}\right) \Delta \cdot (L+\varepsilon)\right] \\
		& = & 1 - \Pr\left[X \leq \left(1- \left(1-\frac{L}{L+\varepsilon}\right)\right) \Delta \cdot (L+\varepsilon)\right] \\
		& \overset{\Cref{lemma:hoeffding}}{\geq} & 1 - exp\left(\frac{-(1-L/(L+\varepsilon))^2 \cdot (\Delta (L + \varepsilon))}{2 \cdot \sigma}\right) \\
		& = & 1 - n^{-\mathfrak{c}} 
	\end{eqnarray*}
\end{proof}

\Cref{lemma:approx_correctness} immediately implies the following corollary:

\begin{corollary} \label{cor:approximation_correctness}
	Let $P$ be fixed for the last $\Delta$ rounds.
	\begin{itemize}
		\item[(i)] If $P \leq L + \varepsilon$ then $P \prec R$ at the server w.h.p.
		\item[(ii)] If $P \geq R - \varepsilon$ then $P \succ L$ at the server w.h.p.
	\end{itemize}
\end{corollary}

We are now ready to show the following lemma:

\begin{lemma}[Convergence/Closure for $P$] \label{lemma:cumulative_stabilization}
	The following statements hold:
	\begin{itemize}
		\item[(i)] Eventually $P \in (L+\varepsilon,R-\varepsilon)$.
		\item[(ii)] Once $P \in (L+\varepsilon,R-\varepsilon)$ it remains in $(L+\varepsilon,R-\varepsilon)$ w.h.p.
	\end{itemize}
\end{lemma}

\begin{proof}
	We first show $(i)$: We know by \Cref{lemma:approx_correctness} that $P$ monotonically increases w.h.p. if $P < L - \varepsilon$, $P$ monotonically decreases w.h.p. if $P > R + \varepsilon$.
	We have to show that we leave states where $P$ is within $[L-\varepsilon,L+\varepsilon]$ or (analogously) within $[R-\varepsilon,R+\varepsilon]$.
	
	First fix $P = L - \zeta$ for an arbitrary fixed $\zeta > 0$.
	The server either decides $P \prec L$, $L \prec P \prec R$ or $P \succ R$ via its approximation algorithm.
	Obviously only the first decision would be correct.
	We know by \Cref{cor:approximation_correctness} that $\Pr[P \succ R] \leq 1/n^c$, so the server deciding $P \succ R$ does not happen w.h.p.
	Also note that the decision for $L \prec P \prec R$ does not modify $P$, so even if the server makes this decision, we do not lose progress on reaching a value for $P$ within $[L,R]$.
	We show that in a round $t$ with $\delta = 0$, the server chooses the correct decision with at least constant probability: Consider the variable $X$ at the server in round $t$.
	By definition of our protocol the server decides $P\prec L$ if $X/\Delta < L$ holds.
	As $\mathbb{E}[X] = P \cdot \Delta = (L - \zeta) \cdot \Delta$ it follows 
	\[\Pr[P \prec L] \geq \Pr[L \prec P \prec R]\]
	and
	\[\Pr[P \prec L] + \Pr[L \prec P \prec R] = 1-n^{-c}.\]
	This immediately implies $\Pr[P \prec L] \geq \frac{1-n^{-c}}{2}$, which is a constant close to $0.5$ as $n^{-c}$ becomes negligible for $n$ and $c$ high enough.
	This leads to $P$ getting increased such that eventually $P \geq L$ holds.
	As $P$ can be increased by no more than $\hat{p}$ in a single round (we only rise one client probability), we know that $P < R - \varepsilon$ (recall that $|R-L| > \hat{p} + 2\varepsilon$).
	This implies that $P$ does not skip over $[L,R]$ in a single round when starting from $P = L - \varepsilon$.
	Notice that in cases where $P \in [L,L+\varepsilon]$ it still holds that w.h.p. the server will decide either $P \prec L$ or $L \prec P \prec R$.
	One can easily verify that it holds $$Pr[P \prec L] \leq \Pr[L \prec P \prec R],$$ so the probability for the server to choose $P \prec L$ is greater than $0$ for any $P \in [L,L+\varepsilon]$, so $P$ monotonically increases.
	
	Now fix $P = R+\varepsilon$.
	By applying the same argumentation from above we know that the server eventually decides to reduce incoming probabilities.
	As the maximum probability for a client is at most $\hat{p}$ and at most $\sigma$ clients can have their probability reduced by the server, we have that $P$ gets reduced by at most 
	\[\hat{p} \sigma  - \frac{\hat{p} \sigma}{1+1/\sigma} = \frac{\hat{p}}{1+1/\sigma} < \hat{p}.\]
	Thus, reducing $P = R + \varepsilon$ by a value less than $\hat{p}$ implies that $P$ eventually reaches a value within $[L,R]$ without skipping over it, so we are done.
	Following a similar argumentation as above, we can deduce that in cases where $P \in [R-\varepsilon,R]$ holds, $P$ is still monotonically decreasing.
	Putting all the pieces together we get that eventually $P \in (L+\varepsilon,R-\varepsilon)$.

	It remains to show $(ii)$: Assume that we are in a legitimate state.
	Since $P \in (L+\varepsilon,R-\varepsilon)$ it follows that $P$ does not get modified because the server will decide $L \prec P \prec R$ w.h.p. (\Cref{lemma:approx_correctness}).
	Therefore $P$ remains within $L \prec P \prec R$ w.h.p.
\end{proof}

All that is left now is to show convergence and closure for fairness.

\begin{lemma}[Convergence/Closure for Fairness]
\label{lemma:fairness}
	The following statements hold:
    \begin{itemize}
    	\item[(i)] Let $P \in (L+\varepsilon,R-\varepsilon)$ be fixed. Eventually it holds $\sum_{v \in  V}\left(p(v)-\frac{P}{n}\right)^2 \leq \frac{1}{n^c}$ for some constant $c > 0$.
        \item[(ii)] Let $P \in (L+\varepsilon,R-\varepsilon)$ and let $\sum_{v \in  V}\left(p(v)-\frac{P}{n}\right)^2 \leq \frac{1}{n^c}$. Then it holds that $\sum_{v \in  V}\left(p(v)-\frac{P}{n}\right)^2 \leq \frac{1}{n^c}$ in any subsequent state as well.
    \end{itemize}
\end{lemma}

\begin{proof}
	We start by showing $(i)$.
	Assume that $P \in (L+\varepsilon,R-\varepsilon)$ is already fixed.
	Define the potential $\Phi = |p_{\mathit{max}} - p_{\mathit{min}}|$, where $p_{\mathit{max}} = \max\{p(v_1),\ldots,p(v_n)\}$ and $p_{\mathit{min}} = \min \{p(v_1),\ldots,p(v_n)\}$.
	Obviously $\Phi = 0$ if and only if all client probabilities are equal.
	Also it is easy to see that $\Phi$ is never increasing, because by applying the average rule probabilities cannot get higher than $p_{\mathit{max}}$ and also not lower than $p_{\mathit{min}}$.
	Finally note that $\Phi$ is monotonically decreasing: Consider the event that $k$ clients $v_1,\ldots,v_k$ ping the server in round $t$ with at least one of the $p(v_i)$'s being equal to $p_{max}$ and at least one probability being lower than $p_{max}$.
	Taking the average of these $k$ values implies that $p_{max}$ (and thus $\Phi$) reduces in round $t$.
	Notice that it may not necessarily hold that eventually $\Phi$ reduces to $0$, because of the way we round the probabilities when computing new probabilities at the server.
	However, by definition of \Cref{algo:protocol} it holds that eventually the maximum distance between $p_{max}$ and $p_{\mathit{min}}$ will be at most $\frac{1}{2^{b \cdot W}}$ since this is the maximum distance that we allow new probabilities of clients to have when computing the average out of their old probabilities.
	Therefore it eventually holds $$\sum_{v \in  V}\left(p(v)-\frac{P}{n}\right)^2 \leq \sum_{v \in V} \left(\frac{1}{2^{b \cdot W}}\right)^2 \leq \sum_{v \in V} \left(\frac{1}{2^{b \cdot \log n}}\right)^2 = \frac{n}{n^{2\omega}} = \frac{1}{n^{2\omega-1}}.$$
	Thus, choosing $c = 2\omega-1$ suffices to prove convergence.
	
	For closure it is easy to see that the fairness formula remains fixed in legitimate states, because the only modification to probabilities that may occur is when the server rounds the computed average value and moves the values $r_i$ representing the least significant bit to different client probabilities.
\end{proof}

At last we show that the server only needs $O(W + \log n)$ memory in legitimate states.

\begin{lemma} \label{lemma:server_storage}
	Let the system be in a legitimate state.
	Then the server needs at most $O(W + \log n)$ bits to store its internal variables.
\end{lemma}

\begin{proof}
	The values for $\varepsilon, L$ and $R$ are constants, so they can be stored via constant many bits.
	Since $\Delta \in \Theta(\log n)$ it can be stored via $O(\log \log n)$ many bits.
	The value of $X$ can be no more larger than $\Delta \cdot \sigma$, so $O(\log \log n)$ bits suffice.
	By \Cref{app:lemma:table_memory} the server needs $O(\log n)$ bits to store the table that is used to approximate $\log n$ (see \Cref{table:log_n} in \Cref{app:approx_log_n:2}).
	The server needs additional bits (of temporary storage) to store the identifiers of clients that ping each round.
	As at most $\sigma$ of these identifiers have to be stored by the server, $O(W)$ bits suffice.
\end{proof}

\Cref{theorem:self_stabilization} follows from \Cref{lemma:cumulative_stabilization}, \Cref{lemma:fairness} and \Cref{lemma:server_storage}.

\section{Lower Bound} \label{app:lower_bound}

\begin{theorem}[Lower Bound]
\label{theorem:convergence_time:lower_bound}
	 Any self-stabilizing protocol needs $\Omega(p_{\mathit{min}}^{-1} \log n + n)$ rounds to reach a legitimate state w.h.p.
\end{theorem}

We need the following lemma in order to show \Cref{theorem:convergence_time:lower_bound}:

\begin{lemma} \label{lemma:whp_ping}
	Consider a client $v \in V$ with fixed probability $p(v) \in (0,\hat{p}]$.
	After $O(p(v)^{-1} \log n)$ rounds, $v$ has pinged the server at least once w.h.p.
\end{lemma}

\begin{proof}
	Assume w.l.o.g. that $p(v)^{-1} \in \mathbb{N}$ and consider a time frame of $T=p(v)^{-1} \log n$ rounds.
	Define the variable $X_j = 1$ if $v$ pings the server in round $j$ (otherwise $X_j = 0$).
	Obviously it holds $\Pr[X_j = 1] = p(v)$.
	Let $X = \sum_{i=1}^{T} X_i$.
	Then \[\mathbb{E}[X] = p(v)^{-1} \cdot \log n \cdot p(v) = \log n.\]
	Using \Cref{lemma:chernoff}(iii) we get
	\[
		\Pr[X=0]  =  \Pr[X \leq (1-1)\cdot \log n] 
		 \leq  \exp\left(\frac{-(1^2 \cdot \log n)}{2}\right)
		 \leq  n^{-c}
	\]
	for some constant $c > 0$, so we know that $v$ pings at least once w.h.p.
\end{proof}

We are now ready to prove \Cref{theorem:convergence_time:lower_bound}:

\begin{proof}[Proof of \Cref{theorem:convergence_time:lower_bound}]
	Since the server does not know the clients that are connected to it, it is only able to modify the probability of a client $v$, once $v$ has sent at least one ping message to the server (and thus also told the server about its reference).
	Therefore any protocol needs each client to ping the server at least once successfully in order to be able to converge, since initially no client may have the correct probability.

	Assume that $p_{\mathit{min}}$ is equal to some constant $c \in (0,\hat{p}]$.
	Then the time until each client has successfully pinged the server at least once is equal to $\Omega(n/\sigma) = \Omega(n)$ rounds, because in each round the server is able to receive at most $\sigma$ pings.
	
	Now assume that $p_{\mathit{min}} = 1/n^c$ for some constant $c > 0$.
	Then the client with probability $p_{\mathit{min}}$ needs $O(p_{\mathit{min}}^{-1} \log n)$ rounds w.h.p. until it sends the first ping message (\Cref{lemma:whp_ping}).
	Assuming an optimal schedule for the client pings (i.e., once a client that did not ping before decides to ping, it will ping successfully), we know by the union bound that each client has pinged the server after $O(p_{\mathit{min}}^{-1} \log n)$ rounds w.h.p..
	This leads to the protocol converging after $\Omega(p_{\mathit{min}}^{-1} \log n)$ rounds w.h.p.
	
	Combining the lower bounds for both cases yields the bound claimed in the theorem.
\end{proof}

\section{Time for \textit{P} to reach \texorpdfstring{$[L,R]$}{[L,R]}} \label{app:P_convergence_time}
In this section we analyze the time it takes until $P$ has converged to some value within $[L,R]$ by proving \Cref{theorem:P_convergence_time}.
For this we have to consider the cases $P < L$ and $P > R$ and analyze the time it takes for $P$ to reach a value within $[L,R]$ in both cases.

For the case $P > R$ we need the following technical lemma.

\begin{lemma} \label{lemma:pre:ping_amount}
	Consider the set $V' \subset V$ such that $V' = \{v \in V\ |\ p(v) \geq \frac{P}{2n}\}$. 
	While $P \geq \Theta(1)$, $O(\log n)$ clients out of the set $V'$ ping the server within $O(\log n)$ rounds w.h.p.
\end{lemma}

\begin{proof}
	For $P \geq \Theta(1)$ it holds that \[\sum_{v\in V'}p(v) > P/2 \in \Theta(1).\]
	Fix $\omega = \frac{2}{1-\exp(-P/2)} \in \Theta(1)$ and consider a time frame of $\omega \cdot \log n$ rounds.
	Let $X_j = 1$ if and only if at least one client $v \in V'$ pings the server in round $j$.
	Then \[\Pr[X_j = 1] \geq 1 - \left(1-\frac{P}{2n}\right)^n = 1 - e^{-P/2}.\]
	Define $X = \sum_{j=1}^{\omega \log n} X_j$, which yields $\mathbb{E}[X] = \omega \log n (1 - e^{-P/2})$.
	Now choose $\delta = 1/4$, which results in
	\begin{eqnarray*}
		\Pr\left[X \leq (1-\delta) \cdot \mathbb{E}[X] \right] & \overset{\Cref{lemma:chernoff}(iii)}{\leq} & \exp\left(\frac{-\delta^2 \cdot \omega \log n \cdot \left(1 - e^{-P/2}\right)}{2}\right)\\
		& = & \exp(-\log n) \\
		& \leq & n^{-c}
	\end{eqnarray*}
	for a constant $c > 0$.
	This implies that the probability that less than $O(\log n)$ clients ping the server in $\omega \cdot \log n = O(\log n)$ rounds is negligible, so $(ii)$ holds as well.
\end{proof}

In order to simplify the analysis, we propose the following extension for our protocol: We store a flag $f \in \{-1,0,1\}$ at the server which is set to $1$ (if $P \succ R$) or $-1$ (if $P \prec L$) once there is a round where $\delta = 0$ but no client has pinged in that specific round.
Then the next round when there has been at least one successful ping at the server, the server either applies the reduction technique in case $f = -1$ or it increases the client probabilities in case $f = 1$.
Afterwards the server resets $f$ to $0$.
By doing so we can guarantee that the server is able to modify $P$ every $\Theta(\Delta)$ rounds.
It is easy to see that we do not violate any constrains when using this approach. 

\begin{lemma} \label{lemma:convergence_time_1}
	Let $P < L$.
	After $O(p_{\mathit{min}}^{-1} \log n)$ rounds $P \in [L,R]$ w.h.p.
\end{lemma}

\begin{proof}	
	Since $L \in \Theta(1)$, there exists a constant $\alpha \in \mathbb{N}$ such that $\alpha \cdot \hat{p} \geq L$.
	Thus we have to wait until $\alpha$ different clients have their probability set to $\hat{p}$ by the server.
	Consider the set of clients that have a reasonably small probability, i.e., the set $S = \{v \in V\ |\ p(v) \leq \hat{p}/2\}$.
	Setting one of these client's probability to $\hat{p}$ increases $P$ in such a way that, for a constant $c_1 > 0$, at most $c_1 \cdot \alpha$ of these events suffice to obtain $P \geq L$.
	It holds $|S| = \Theta(n)$ because for $|S| < \Theta(n)$ we would have already obtained $P \geq L$.
	Assume $|S| = n/c_2$ for a constant $c_2 > 0$.
	Since each client out of $S$ has probability at least $p_{\mathit{min}}$, we get the following bound for the probability that no client out of $S$ pings in a round where $\delta = 0$: \[\Pr[\text{No } v \in S \text{ pings in a round where } \delta = 0] \leq (1-p_{\mathit{min}})^{n/c_2}.\]
	Now consider $O(p_{\mathit{min}}^{-1})$ rounds in which $\delta = 0$ and let $A$ be the event that no client out of $S$ pings in any of these rounds.
	We get \[\Pr[A] \leq (1-p_{\mathit{min}})^{(n/c_2) \cdot p_{\mathit{min}}^{-1}} \leq e^{-n/c_2}.\]
	Thus, w.h.p., at least one client out of $S$ pings within $O(p_{\mathit{min}}^{-1})$ rounds where $\delta = 0$.
	Also, the probability for such a ping to be successful is at least constant, since $P < L < \sigma$.
	As $\delta = 0$ occurs every $\Delta = \Theta(\log n)$ rounds, the lemma follows.
\end{proof}

\begin{lemma} \label{lemma:convergence_time_2}
	Let $P > R-\epsilon$. Then the following two statements hold:
	\begin{enumerate}
	    \item After $O(n \log^2 n)$ rounds, it holds $P \in [L,R]$ w.h.p.
	    \item 
	\end{enumerate}{}
	
\end{lemma}

\begin{proof}	
	Fix $P > R - \epsilon$.
	Due to \Cref{lemma:pre:ping_amount} we can conclude that, with at least constant probability, at least one client out of the set $V' = \{v \in V\ |\ p(v) \geq \frac{P}{2n}\}$ pings the server successfully in a round where $\delta = 0$ at the server \emph{and} the server approximates $P \succ R$.
	
	This implies that $P$ is reduced by at least \[\frac{P}{2n} - \frac{P}{2n(1+1/\sigma)} = \frac{P}{2n (1 + \sigma)}\] 
	That means, in the next round $P$ is equal to $P-\frac{P}{2n(1+\sigma)} = (1-\frac{1}{2n(1+\sigma)}) \cdot P$.
	
	Let now $\varphi$ be a lower bound for the probability for a reduction.
	As one reduction occurs every $\Delta$ rounds, after $T\Delta$ rounds, the expected value $P$ is at most $$\left(1-\frac{\varphi}{2n(1+\sigma)}\right)^T \cdot P.$$
	
	Now we can proof the two statements:
	\begin{enumerate}
	\item As long as $P > R$, it holds $\varphi \geq \frac{1}{2}$. 
	Setting $T = 4n(1+\sigma)\log n = O(n \log n)$ yields $P \in \Theta(1)$, so $O(n \log n)$ reductions suffice.
 These $O(n \log n)$ reduction can be achieved in $O(n \log^2 n)$ rounds w.h.p.
	\item Suppose that $\varphi > \frac{1}{n}$ and $P \leq R$.
	As long as this is the case, it holds:
	$$\left(1-\frac{1}{2n^2(1+\sigma)}\right)^T \cdot P.$$
	Thus, setting $T = 4n^2(1+\sigma)\log n = O(n \log n)$ yields $P \leq \frac{R}{e} \leq R - \epsilon$.
	However $P \leq R- \epsilon$ implies that $\varphi \geq 1 - o(\frac{1}{n})$.
	Since this is a contradiction, the statement follows.
	\end{enumerate}{}
	This concludes the lemma.
\end{proof}

The combination of \Cref{lemma:convergence_time_1} and \Cref{lemma:convergence_time_2} implies \Cref{theorem:P_convergence_time}.

\section{Fairness in the Legitimate State} \label{app:fairness_complete}
\label{sec:optimization}

In the last step, we will bound the time until all probabilities are \emph{almost} the same.
To be precise, the pairwise difference between probabilities will be $O\left(\frac{1}{n}\right)$ w.h.p. 
Note that in this phase, the sum of all probabilities will not change anymore w.h.p.
Thus, the algorithm will only average the probabilities of all nodes that successfully ping in a given round.

We begin our analysis with the observation that with constant probability, the average of all received values is within the magnitude of the arithmetic means, i.e.,

\begin{lemma} \label{lemma:fairness:constant_prob}
Let $\mathcal W \subset V$ be any set of nodes and let $P_{\mathcal W} := \sum_{w \in W} p_t(w) \leq \sigma$ be the sum and $M := \frac{P_{\mathcal W}}{|\mathcal W|}$ the arithmetic mean of its probabilities.
Let $S(t)$ be the average of all probabilities received by the server in round $t$.
Then it holds:
$$
	\Pr\left[S(t) \geq \frac{M}{8}\right] \geq 1- {e^{-c}}
$$
With $c \in O(P_{\mathcal W})$ being a value that only depends on $P_{\mathcal W}$.
\end{lemma} 
\begin{proof}

For the proof, we divide the set of nodes in \emph{good} nodes $\mathcal{G} \subset \mathcal W$ and \emph{bad} nodes $\mathcal{B} \subset \mathcal W$. 
We call nodes with $p_t(v) \geq \frac{M}{4}$ \emph{good}, all others are bad.
Intuitively, we wish to lower bound the probability that \emph{most} of the nodes that successfully ping are good.
The analysis is complicated by the fact that the random experiment works in two steps. 
In the first step, all nodes independently send their probabilities to the server. 
In the second step, we uniformly at random pick at most $\sigma$ nodes that sent their probabilities.
Since this corresponds to drawing without replacement and is also highly dependent on the first phase, 
the experiment is not independent and we need to observe it more carefully.

We will first concentrate on the first step.
Therefore, we observe the following two independent events $\mathcal{E}_1$ and $\mathcal{E}_2$.
$\mathcal{E}_1$ denotes the event that at most $\frac{P_{\mathcal{W}}}{2}$ bad nodes send their probabilities to the server.
Analogously, $\mathcal{E}_2$ denotes the event that at least $\frac{P_{\mathcal{W}}}{2}$ good nodes send their probabilities to the server.
In the following we will show that it holds
$$
	\Pr[\mathcal{E}_1 \cap \mathcal{E}_2] \geq 1-e^{O(P_{\mathcal W})}
$$
Since the events are independent, we can observe them individually. 
In both cases, we let $X_v \in \{0,1\}$ be the RV that denotes if $v \in V$ pings the server. 
Further, we denote $X_{\mathcal{G}} := \sum_{v \in \mathcal G} X_v$ and $X_{\mathcal{B}} := \sum_{v \in \mathcal B} X_v$ be number of sending good and bad nodes, respectively.

\begin{enumerate}
    \item First, consider the expected number of bad nodes that send a message. 
    Since there can be at most $n-1$ bad nodes (otherwise all nodes would be bad, which is a contradiction) and every bad node has a probability of at most $\frac{P}{4n}$ (per 	definition) the expected number of bad nodes that send can be bounded as follows:
    $$
        \mathbb{E}[X_{\mathcal{B}}] \leq \sum_{i=1}^{n-1} \frac{P}{4n} \leq \frac{P}{4}
    $$
    Since all bad nodes send independently and the corresponding variables are binary, we can apply the Chernoff Bound. 
    This yields
    $$
    	\Pr\left[X_\mathcal{B} \geq \frac{P}{2}\right] = \Pr\left[X_\mathcal{B} \geq (1+1)\mathbb{E}[X_\mathcal{B}]\right] \leq exp\left(-\frac{\mathbb{E}[X_\mathcal{B}]}{3}\right) \leq exp\left(-\frac{P}{12}\right)
    $$
    Recall that $P_{\mathcal W}$ is a (probably small) constant in $O(1)$ since $P_{\mathcal W} \in O(\sigma)$ per assumption and we assume that $\sigma$ is some constant.
    Thus, with (at least constant) probability at least $1-exp\left(-\frac{P_{\mathcal W}}{12}\right)$ no more than $\frac{P_{\mathcal W}}{2}$ bad nodes send.
    \item Second, we observe the probability that the sum of all good nodes that sent is at least $O(\frac{\sigma^2}{n})$. 
Recall that the sum of all probabilities is $P_{\mathcal{W}}$ and thus, it holds:
$$
	\mathbb{E}\left[\sum_{v \in W} X_v \right] = P_{\mathcal{W}}
$$
Furthermore, it holds:
$$
	\mathbb{E}\left[\sum_{v \in {\mathcal{W}}} X_v \right] = \mathbb{E}\left[\sum_{v \in \mathcal{G}} X_v \right] + \mathbb{E} \left[\sum_{v \in cB} X_v \right] 
$$
This follows from the fact that each node is either good or bad.
If we rearrange this, we get:
$$
	\mathbb{E}\left[\sum_{v \in \mathcal{G}} X_v \right] = \mathbb{E} \left[\sum_{v \in V} X_v \right] - \mathbb{E} \left[\sum_{v \in \mathcal B} X_v \right] \geq \frac{3}{4}\mathbb{E} \left[\sum_{v \in V} X_v \right] = \frac{3 P_{\mathcal{W}}}{4}
$$
We will now apply the Chernoff bound again, but this time we use it to obtain a lower bound. 
It holds:
\begin{eqnarray*}
   \Pr \left[X_\mathcal{B} \leq \frac{P}{2} \right] & = & \Pr \left[X_\mathcal{G} \geq (1-\frac{1}{3}) \cdot \mathbb{E}[X_\mathcal{G}] \right] \\
   & \overset{\Cref{lemma:chernoff}(iii)}{\leq} & exp\left(-\frac{3 \cdot \frac{1}{9}\mathbb{E}[X_\mathcal{G}]}{8}\right) \\   
   & \leq & exp\left(-\frac{P}{24}\right)
\end{eqnarray*}
\end{enumerate}
Now recall that $\mathcal{E}_1$ and $\mathcal{E}_2$ are independent. 
Therefore, it holds:
\begin{eqnarray*}
   \Pr[\mathcal{E}_1 \cap \mathcal{E}_2] & \geq & \left(1 - exp\left(-\frac{P_W}{12}\right)\right) \cdot \left(1 - exp\left(-\frac{P_W}{24}\right)\right) \\
   & \geq & \left(1 - exp\left(-\frac{P_W}{24}\right)\right)^2 \\   
   & \geq & \left(1 - exp\left(-\frac{P_W}{24\log(2)}\right)\right)^2
\end{eqnarray*}

This was to be shown.

Now analyze the second phase under the condition that event $\mathcal{A} := \mathcal{E}_1 \cap \mathcal{E}_2$ occurred.
Only here, we need to consider the fact that at most $\sigma$ of all nodes that sent their probability are actually considered by the server.
Per definition, we assume that $P_{\mathcal W} \leq \sigma$ and thus obviously $\frac{P_{\mathcal W}}{2}<\frac{\sigma}{2}$.
Since fewer than $\frac{P_{\mathcal W}}{2}$ bad nodes have sent their probability, 
at most half of all nodes drawn in the second phase are bad.
The other half must be good since more than $\frac{P_{\mathcal W}}{2}$ good nodes sent.
Since every good node has at least a probability of $\frac{M}{4}$ the average probability must therefore be at least $\frac{M}{8}$. 
This was to be shown.
\end{proof}
Using this insight, we can bound the time until all nodes $v \in V$ have $p(v) \in O(\frac{P}{n})$ via a simple potential function.

Before we begin with the proof, we need following technical lemmas that will simplify the potential analysis.
\begin{lemma}
\label{lemma:helper-square}
Let $x,y \in \mathbb{Z}$ be two integer values, then it holds
$$
	x^2 + y^2 - \left( \left( \lfloor \frac{x+y}{2}\rfloor \right)^2+ \left( \lceil \frac{x+y}{2}\rceil \right)^2 \right) \geq \frac{(x-y)^2 - 1}{2} 
$$
\end{lemma}
\begin{proof}
Suppose $x+y$ is odd, otherwise it holds $\lceil \frac{x+y}{2}\rceil = \lfloor \frac{x+y}{2}\rfloor$.
Observe that it holds
$$
	\lfloor \frac{x+y}{2}\rfloor = \frac{x+y}{2} - \frac{1}{2}
$$
and analogously
$$
	\lceil \frac{x+y}{2}\rceil = \frac{x+y}{2} + \frac{1}{2} 
$$
Thus, through application of the first and second binomial law we get:
$$
 \left( \left( \lfloor \frac{x+y}{2}\rfloor \right)^2+ \left( \lceil \frac{x+y}{2}\rceil \right)^2 \right) = 2\left(\frac{x+y}{2}\right)^2 + 2\left(\frac{1}{2}\right)^2
$$
Therefore, the whole term simplifies to
$$
	x^2 + y^2 - 2\left(\frac{x+y}{2}\right)^2 + 2\left(\frac{1}{2}\right)^2
$$
Evaluating the second term using the second binomial law and a subsequent simplification gives us:
$$
	\frac{(x-y)^2 - 1}{2} 
$$
This was to be shown.
\end{proof}

\begin{lemma}
\label{lemma:helper-distance}
Let $x_1, \dots, x_n$ be values such that $\sum_{i=1}^n x_i^2 := y$ and $\sum_{i=1}^n x_n = 0$ 
$$
	\sum_{i=1}^{n-1} \sum_{j=1}^n (x_i-x_j)^2 := 2ny
$$
\end{lemma}
\begin{proof}
Using the second binomial law, we get:
$$
	\sum_{i=1}^{n} \sum_{j=1}^n (x_i-x_j)^2 = \sum_{i=1}^{n} \sum_{j=i}^n x_i^2 + x_j^2 - 2x_ix_j := \left(\sum_{i=1}^{n} \sum_{j=1}^n x_i^2 + x_j^2\right) - \sum_{i=1}^{n} \sum_{j=i}^n 2x_ix_j 
$$
Now observe the first term. 
Each $x_i^2$ appears exactly $2n$ times in the sum.
To be precise, $n$ times as the first summand and $n$ times as the second summand.
Thus, it can be rewritten as:
$$
	\sum_{i=1}^{n} \sum_{j=i}^n x_i^2 + x_j^2 := 2n-1 \sum_{i=1}^n x_i^2 := 2ny
$$
It remains to bound the second term.
Here, it holds:
$$
	\sum_{i=1}^{n} \sum_{j=1}^n 2x_ix_j := 2\left(\sum_{i=1}^n\sum_{i=1}^n x_ix_j\right) = 2\left(\sum_{i=1}^n x_i\right)^2 :=0
$$
Hence, the lemma follows.
\end{proof}

\begin{lemma} \label{lemma:fairness:convergence_time}
Let $p_{min}$ be lowest probability in the system and let $\mu \in [0,\frac{P}{16n}]$ be any threshold. 
Further, suppose that the probability that node with probability smaller than $2\mu$ is decreased within $O(p^{-1}_{min}\log{n})$ rounds is $p_d \in [0,1]$.
Then all probabilities are within $O(\mu)$ after $O((p_dp_{min})^{-1}\log^2{n})$ rounds w.h.p.
\end{lemma}

\begin{proof}
We prove the statement via a potential function.
First, let $d_{\mu}: (0,1) \to (0,1)$ for some $\mu \in (0,1)$ be defined as follows:
$$
	d_{\mu}(x) = \begin{cases}
		\mu -  x & \textit{if } x \leq \mu\\
		0 & \textit{otherwise}
	\end{cases}
$$  
Then we define the potential as follows:
$$
	\Phi(t) := \sum_{v \in V} d_{\mu}(2^{kW}p_t(v))^2
$$
In the remainder, we denote the differences as $d_v(t) := d_\mu(2^{kW}p_t(v))$ for short.
Since all probabilities are within $[\frac{1}{2^{kW}},1]$ we effectively observe integer values $1, \dots, 2^{kW}$.
Therefore, we can use the well-defined notions of $\lceil\rceil$ and $\lfloor\rfloor$ to simplify notation and calculations.
Last, note that any value $d_v(t) \leq 1$ implies that $p_v(t)$ is smaller than $\frac{\mu}{2^{kW}}$.

The proof's idea is simple.
The only thing that increase the potential is a reduction.
We condition on the fact that no reduction happens, and analyze the process
in absence of reductions, i.e., only balancing are applied.
Whenever, a reduction is applied, we \emph{restart} the analysis.

The actual proof now has three steps:
First, we make ourselves clear that $\Phi(t)$ decreases if the algorithm balances two (or more) probabilities.
Therefore, we will use \Cref{lemma:helper-square}.
This also implies that the potential can only decrease if the algorithm applies reduction by $(\frac{1}{1-\frac{1}{\sigma}})$.
Second, we bound the expected change in potential if $\mu := \frac{P}{16n}$.
Last, we use the Markov inequality to show that after $O(p_{\mathit{min}}^{-1}\log n)$ rounds, all nodes probabilities are within $O(\mu)$.
This proves the lemma.

We will now show each of these steps separately.
\begin{enumerate}
\item We will begin by showing that $\Phi(t)$ can only decrease.
Before we start the actual calculations, we make the following observations:
\begin{enumerate}
\item We only need to consider the case that the algorithm builds the average of two nodes $v$ and $w$.
If a set of two or more nodes are balanced, we can decompose it into an infinite series of pairwise balances with the same result.
Thus, if every pairwise balancing is monotone, the $\sigma$-wise balancing must be monotone too.
\item Since the balancing only affects $v$ and $w$, we only need to compare $d_v(t)^2$ and $d_w(t)^2$.
In particular, it suffices to show that $d_v(t)^2 + d_w(t)^2 \geq d_v(t+1)^2 + d_w(t+1)^2$
\item If $d_v(t) - d_w(t) = 0$ both probabilities are below $\mu$ and their average certainly cannot be bigger than $\mu$. Thus, the potential remains unchanged. 
\item If $(d_v(t) - d_w(t))^2 = 1$ the average is either smaller than $\mu$ or the excess bit will simply swap its position. Thus, the potential remains unchanged.
\end{enumerate}
All in all, we will consider the case that $v$ and $w$ are balanced and $d_v(t) + d_w(t) \geq 1$.
Then $d_v(t+1)^2 + d_w(t+1)^2$ can be simplified as follows:
\begin{align*}
&d_v(t+1)^2 + d_w(t+1)^2 \\
&\leq \left( \mu - \lfloor\frac{\mu+d_v(t) + \mu+d_w(t)}{2}\rfloor \right)^2 + \left(\mu - \lceil \frac{\mu+d_v(t) + \mu+d_w(t)}{2}\rceil \right)^2\\
&\leq \left( \lfloor\frac{d_v(t) + d_w(t)}{2}\rfloor \right)^2 + \left(\lceil\frac{d_v(t) + d_w(t)}{2}\rceil \right)^2\\
\end{align*}
Thus, we can apply \Cref{lemma:helper-square} and get:
\begin{align*}
	& d_v(t)^2 - d_w(t)^2 - (d_v(t+1)^2 + d_w(t+1)^2)\\
	& = d_v(t)^2 - d_w(t)^2- \left( \lfloor\frac{d_v(t) + d_w(t)}{2}\rfloor \right)^2 + \left(\lceil\frac{d_v(t) + d_w(t)}{2}\rceil \right)^2\\
	& \geq \frac{1}{2}(d_v(t)-d_w(t))^2 - \frac{1}{2}
\end{align*}
Since $d_v(t)-d_w(t) > 1$ by assumption, the term is bigger $0$.
This was to be shown.
\item This step has two substeps.
First, we observe we observe the expected decrease in potential given that the algorithm performs a balancing.
Second, we observe the increase through a reduction of probabilities.
\begin{enumerate}
\item Condition on the fact that the algorithm performs a balancing in round $t$ and call this event $B_t$.
Next, we define a set of good events $\mathcal{A}_v$ with $v \in V$ that decrease the potential by at least $d_v(t)^2$.
Since the potential is monotone, there is no event that increases the potential.
Thus, it suffices to observe this subset of events.

Consider any node $v \in V$ with $d_v(t)>0$. 
First, we bound the probability that $v$ successfully pings the server.
Therefore, $v$ must send a message in the first step and then be picked by the server in the second step.
By the Markov inequality, the probability that more than $2P$ nodes send a message is at most $\frac{1}{2}$.
Thus, with probability $\frac{1}{2}$ or more, less than $2P$ nodes send.
Since $2P \leq 2\sigma$ every node that pinged, is picked with probability at least $\frac{1}{2}$.
Thus, the probability that $v$ pings and is then picked is lower bounded by $\frac{p_{\mathit{min}}}{4}$.

Now we condition on the event that $v$ successfully pinged.
To raise $p_t(v)$ above $\mu$ and thus reducing the potential, 
the average probability of all other nodes that successfully ping must be (at least) $\frac{\mu}{2}$.
The probability for this event can be calculated via \Cref{lemma:fairness:constant_prob}. 
We simply set ${\mathcal{W}} = V \setminus \{v\}$ and $\sigma' := \sigma+1$.
Then, with constant probability $1-e^{O(P)}$ the average of all other nodes that successfully send is at least $\frac{P}{8}$.
Thus, $p_v(t)$ must be set to some value below $\frac{P}{16}$.
This implies $d_v(t+1)=0$ and thus changes the potential by (at least) $d_v(t)^2$.
We can define such a good event $\mathcal{A}_v$ for all nodes $v \in V$ with $d_v(t)>0$.
\item Condition on the fact that the algorithm performs a decrease in round $t$ and call this event $D_t$.
Similar to the good events in the first step, we will now define \emph{bad} events that increase the potential.
Note that probability that a node pings exactly in a round where the probabilites are decreased is $\frac{1}{\Delta} \in \Theta(\frac{1}{\log{n}})$.
Thus, the expected increase for a  node $v$ is $O(\frac{d(v)^2}{\log{n}})$ given that $v$ send successfully. 
\end{enumerate}
The expected change in potential $\Delta(t+1) := \Phi(t)-\Phi(t+1)$ is therefore at least:
\begin{align*}
	\mathbb{E}[\Delta(t+1)] &\geq \sum_{v \in V} p_t(v) \left(\Pr[B_t]\mathbb{E}[\Delta(t+1)|B_t] + \Pr[D_t]\mathbb{E}[\Delta(t+1)|D_t] \right)\\
	&\geq \sum_{v \in V: d_v(t)>0} p_t(v) \left( c_1 d_v(t)^2 - \frac{(1+\frac{1}{\sigma})}{\log n}d_{\mu}(v)^2 \right)\\
	&= \sum_{v \in V: d_v(t)>0}p_t(v) (1-\frac{1}{\log n})(1-\frac{1}{e^{O(P)}}) d_v(t)^2 \geq p \cdot \Phi(t)  
\end{align*}
Here, $p \in O((p_{min}+\frac{1}{n^2})\log{n}^{-1})$ captures all constant factors and the lowest value $p_v(t)$ can obtain.
\item The previous observation implies:
$$
	\mathbb{E}[\Phi(t+1)\ |\ \Phi(t)] = \Phi(t) - \mathbb{E}[\Delta(t+1)] \leq \Phi(t) - p \cdot \Phi(t) = (1 - p) \cdot \Phi(t) 
$$
And furthermore by induction:
$$
	\mathbb{E}[\Phi(t+T)\ |\ \Phi(t)] \leq (1 - p)^T \cdot \Phi(t) 
$$
Thus, after $T := c_1\cdot p^{-1}\log n \in O(p_{\mathit{min}}^{-1}\log n)$ rounds the expected value is:
$$
    \mathbb{E}[\Phi(t+T)] \leq (1-p)^T \cdot \Phi(t) \leq \frac{\Phi(t)}{n^{c_1}}
$$
Now the potential is maximal if each $d_v(t)^2$ is maximal. 
Thus, the maximal value for $\Phi(t)$ is $n\left(\frac{P}{16n}\right)^2$.
Furthermore $P$ can be at most $n$ because otherwise there would be node with probabilities above $1$.
Choosing $c_1>5$ and $c_2 := c_1 - 4$ thus yields:
$$
	\mathbb{E}[\Phi(t+T)] \leq \frac{1}{n^{c_2+1}}
$$ 
Last, through Markov we get:
$$
	\Pr \left[\Phi(t+T) \geq \frac{1}{n} \right] \leq \frac{1}{n^{c_2}}
$$
\end{enumerate}
This, after $O(p_{min}\log{n})$ rounds without reduction the all probabilities are within $\Omega(\mu)$.
Since the probability that reduction happens within this time is $p_d$, 
we need to repeat this is experiment $O(p_d\log n)$ times until there is 
an execution without reduction.
This proves the lemma.
\end{proof}

At last we also show \Cref{theorem:deckel}:

\begin{proof}[Proof of \Cref{theorem:deckel}]
This proves for the most part follows a the proof of Lemma $2$ in \cite{berenbrink}. 
We only need to make small adaption to account for our non-uniform sending probabilities and non-pairwise balancings.
We again use a potential function, the expected change in potential, and Markov inequality to get to the desired bound. 
As the potential, we now use the $L_2^2$ distance to the arithmetic mean $\varnothing := \frac{P}{n}$. 
First, we define difference function $d(x)$ as 
$$
	d(x) := (x-\varnothing)
$$
As before, we use $d_v(t)$ as shorthand notation for $d(p_v(t))$.
Then the potential is:
\begin{equation}
  \Phi(t)
= \sum_{i=1}^{n} {\left( d_v(t) \right)}^2
\end{equation}
The modus operandi is the same as in the proof of \Cref{lemma:fairness:convergence_time}.
We start by showing that $\Phi$ is monotonically decreasing.
Then, we bound the expected change of the potential if two nodes are paired.
We extend this to an arbitrary set.
The analysis is concluded by an application of Markov's inequality to prove that
w.h.p. the true potential does not deviate too far from its expectation.
 
\begin{enumerate}
\item We again only observe the difference if two nodes are balanced.
As before any balancing action that involves more than two nodes can be decomposed into an infinite series of pairwise balances.
If each individual balancing is monotone, any sequence must be too.
Again, we consider the balancing of $v$ and $w$ and the corresponding change of $d_v(t)^2$ and $d_w(t)^2$.
Here, it holds:
$$
d_v(t)^2 + d_w(t)^2 - \left( d_v(t+1) + d_w(t+1) \right) \geq \frac{(d_v(t) + d_w(t))^2 - 1}{2} \geq 0
$$
This results from the same arguments as in \Cref{lemma:fairness:convergence_time} and follows from \Cref{lemma:helper-distance}.
\item  We now calculate the expected change in potential in each step.
Therefore, we consider the following simplified process.
Instead of constructing the average of all probabilities, the server creates a random matching between all probabilities.
Then, the average of each individual pair is computed, rounded, and send to the corresponding nodes.
This lower bounds the change of the potential, as the balancing of all values can be decomposed in an infinite series of pairwise matchings.
For each pair $(v,w) \in V^2$ let $Y(v,w)$ the random variable that the probabilities of $v$ and $w$ are paired.
This happens if both $v$ and $w$ successfully and are then randomly paired.
Both $v$ and $w$ successfully send with probability at least $\frac{1}{8}p_t(v)p_t(w) \in \Omega\left(\left(\frac{P}{n}\right)^2\right)$. The probability that $v$ is also paired with $w$ is then surely $O(\frac{1}{P})$.
To see this, imagine the experiment as follows: 
Suppose that $v$ and $w$ successfully sent.
First, build a random permutation of all active nodes, then pair each even element with the next even node.
Now, condition that no more than $2P$ nodes successfully sent. 
According to Markov's inequality this happens with prob. at least $\frac{1}{2}$.
Given that $v$ at an even position $i$ (which happens with prob. $\frac{1}{2}$) the probability that $w$ is at $i+1$ is is at least $\frac{1}{2P-1}$.
Thus, chaining all these events gives us a probability of at least $\frac{1}{16P}$.

Summarizing all these observations yield that the probability that $v$ and $w$ are paired is at least $c_1 P \frac{1}{n^2}$ for any small but constant $c_1 \leq \frac{1}{128}$.
This allows us to bound the expected value as follows.
\begin{align*}
\mathbb{E}[\Delta(t+1)| \Phi(t) = \phi] &\geq \sum_{v,w \in V^2} \mathbb{E}[Y(v,w)]\\
&\geq \sum_{v,w \in V^2} \left(c_1 P \frac{1}{n^2}\right) \left(d_v(t)^2 + d_w(t) - \frac{1}{2}\right)\\
&= c_1 P \frac{1}{n^2} \sum_{v,w \in V^2}  d_v(t)^2 + d_w(t) - \frac{1}{2}\\
&= c_1 P \frac{1}{n} \phi - \frac{1}{2}\\
\end{align*}
The last step followed from \Cref{lemma:helper-distance} with implies that $\sum_{v,w \in V^2}(d_v(t)-d_w(t)) = n\phi$.

\item By induction, the expected potential after $T$ rounds can now be bounded as:
$$
	\mathbb{E}\left[\Phi(t+T)|\Phi(t)=\phi\right] \leq \left( 1 - c_1 P \frac{1}{n}\right)^T\phi + \frac{T}{2}
$$
For any $T \geq c_2 \frac{n}{P} \cdot \log n$ with $c_2 \geq 128(k+1)$ the first term becomes negligible and it holds:
$$
	\mathbb{E}\left[\Phi(t+T)|\Phi(t)=\phi\right] \leq \frac{c_2n \cdot \log n}{P}
$$
Now, the probability that we greatly derive from this can be easily bound through Markov:
$$
	\Pr\left[\Phi(t+T) \geq n^{k-1}|\Phi(t)=\phi\right] \leq \frac{P}{n^{k-2}} \leq \frac{1}{n^{k-3}}
$$
\end{enumerate}

\end{proof}

\end{document}